\documentclass[11pt,english]{article}

\usepackage{amsmath,amssymb,amsfonts}
\usepackage{paralist}
\usepackage{amsthm}
\usepackage{fullpage, comment}
 \usepackage{mathtools}
\usepackage{color}
\usepackage{thm-restate}

\newtheorem{theorem}{Theorem}[section]
\newtheorem{lemma}[theorem]{Lemma}
\newtheorem*{lemma*}{Lemma}
\newtheorem{proposition}[theorem]{Proposition}
\newtheorem{lem}[theorem]{Lemma}
\newtheorem{fact}[theorem]{Fact}
\newtheorem{corollary}[theorem]{Corollary}

\newtheorem{claim}[theorem]{Claim}

\newtheorem{prop}[theorem]{Proposition}
\newtheorem{definition}{Definition}

\newtheorem{question}{Question}

\theoremstyle{remark}
\newtheorem{remark}[theorem]{Remark}

\newcommand{\F}{\ensuremath{\mathbb{F}}}

\newcommand{\R}{\ensuremath{\mathbb{R}}}

\newcommand{\Z}{\ensuremath{\mathbb{Z}}}

\newcommand{\round}[1]{\lfloor{#1}\rceil}

\newcommand{\mmod}{~\mathrm{mod}~}
\newcommand{\inprod}[2]{\langle #1, #2 \rangle}

\newif\ifnotes\notestrue
\ifnotes
\newcommand{\enote}[1]{\textcolor{red}{{\bf (Elena:} {#1}{\bf ) }} \marginpar{\tiny\bf
             \begin{minipage}[t]{0.5in}
               \raggedright ELENA NOTE
            \end{minipage}}}
\newcommand{\mnote}[1]{\textcolor{red}{{\bf (Mahdi:} {#1}{\bf ) }} \marginpar{\tiny\bf
             \begin{minipage}[t]{0.5in}
               \raggedright MAHDI
            \end{minipage}}}            
  
\newcommand{\knote}[1]{\textcolor{red}{{\bf (Karthik:} {#1}{\bf ) }} \marginpar{\tiny\bf
             \begin{minipage}[t]{0.5in}
               \raggedright ??
            \end{minipage}}}
\else
\newcommand{\vnote}[1]{}
\newcommand{\knote}[1]{}
\newcommand{\mnote}[1]{}
\newcommand{\enote}[1]{}
\fi

\newcommand{\eps}{\epsilon}

\newcommand{\calz}{\mathcal{Z}}

\newcommand{\etal}{\textit{et al.}}

\newcommand{\var}{\text{var}}
\newcommand{\proj}{\text{proj}}

\newcommand{\ignore}[1]{}
\newcommand{\sidenote}[1]{ \marginpar{\tiny\bf
             \begin{minipage}[t]{0.5in}
               \raggedright *
            \end{minipage}}}

\usepackage{hyperref}

\title{Local Testing for Membership in Lattices}

\author{
Karthekeyan Chandrasekaran \thanks{Department of Industrial and Enterprise Systems Engineering, University of Illinois Urbana-Champaign, IL. Email: {\tt karthe@illinois.edu}.}
\and
Mahdi Cheraghchi \thanks{Department of Computing, Imperial College London, UK. Work done in part while the author was with Simons Institute for the Theory of Computing, University of California, Berkeley, CA and supported by a Qualcomm fellowship. Email:
{\tt m.cheraghchi@imperial.ac.uk}.} 
\and
Venkata Gandikota\thanks{Department of Computer Science, Purdue University,  West Lafayette, IN. Email: {\tt vgandiko@purdue.edu}. Research supported in part by the Purdue Research Foundation.}
\and Elena Grigorescu
    \thanks{Department of Computer Science, Purdue University,  West Lafayette, IN. Email: {\tt elena-g@purdue.edu}.  Research supported in part by the Purdue Research Foundation.}
}

\begin{document}
\maketitle


\begin{abstract}

Motivated by the structural analogies between {\em point lattices} and {\em linear error-correcting codes}, and by the mature theory  on {\em locally testable codes}, we initiate a systematic study of   local testing for membership in  lattices. Testing membership in lattices is also motivated in practice, by applications to integer programming, error detection in lattice-based communication, and cryptography. 

Apart from establishing the conceptual foundations of lattice testing, our results include the following:

\begin{enumerate}
\item We demonstrate upper and lower bounds on the query complexity of local testing 
for the well-known family of {\em code formula} lattices. 
Furthermore, we instantiate our results with code formula lattices constructed from Reed-Muller codes, and obtain nearly-tight  bounds.

\item We show that in order to achieve low query complexity, it is sufficient to design one-sided non-adaptive {\em canonical} tests. This result is akin to, and based on an  
analogous result for error-correcting codes due to Ben-Sasson \etal\  (SIAM J. Computing 35(1) pp1--21).
\end{enumerate}
\end{abstract}

\section{Introduction}
\label{sec:intro}
Local testing for properties of combinatorial and algebraic objects have widespread applications and have been intensely investigated in the past few decades. The main underlying goal in Local Property Testing is to distinguish objects that satisfy a given property from objects that are far from satisfying the property, using a small number of observations of the input object. 
Starting with the seminal works of \cite{BLR93, ref:FS95, RS96}, significant focus in the area has been devoted to locally testable error-correcting codes, called Locally Testable Codes (LTCs) \cite{Goldreich10}. LTCs are the key ingredients in several fundamental results in complexity theory, most notably in the PCP theorem \cite{ALMSS98, AS98}.

In this work we initiate the study of local testability for membership in {\em point lattices},   
a class of infinite algebraic objects that form subgroups of $\Z^n$. 
Lattices are well-studied in mathematics, physics and computer science due to their rich algebraic structure \cite{conway99}. Algorithms for various lattice problems have directly influenced the ability to solve integer programs \cite{Eisenbrand03, Lenstra83, Kannan87}. 
Recently, lattices have found applications in modern cryptography due to attractive properties that enable efficient computations and security guarantees \cite{MGbook, Micciancio:LLL25, Regev06,  Regev10}.
Lattices are also used in practical communication settings to encode data in a redundant manner in order to protect it from channel noise during transmission \cite{ForneyI}.

A point lattice $L\subset \R^n$ of {\em rank} $k$ and {\em dimension} $n$  is specified by a set of linearly independent vectors $b_1,\ldots, b_k\in \Z^n$ known as a basis, for some $k\le n$. If $k=n$ the lattice is said to have {full rank}. The set $L$ is defined to be the set of all vectors in $\R^n$ that are integer linear combinations of the basis vectors, i.e., $L:=\{ \sum_{i=1}^{k} \alpha_i b_i \mid \ \alpha_i\in \Z\ \forall\ i\in[k] \}$.
Lattices are the analogues over $\Z$ of linear error-correcting codes over a finite field $\F$, which are  generated as $\F$-linear combinations of a linearly independent set of basis vectors $b_1,\ldots, b_k\in \F^n$. 

Given a basis for a lattice $L$, we are interested in testing if a given input $t\in \R^n$ belongs to $L$, or is far from all points in $L$ by querying a small number of coordinates of $t$. 
We emphasize that this setting does not limit the computational space or time in pre-processing the lattice as well as the queried coordinates. The main goal is to design a tester that queries only a small number of coordinates of the input. 

\subsection{Motivation}

\noindent{\bf Integer Programming.} Lattices are the fundamental structures underlying integer programming problems. An integer programming problem (IP) is specified by a constraint matrix $A \in \Z^{n\times m}$, a vector $b\in \R^{n}$. The goal is to verify if there exists an integer solution to the system $Ax=b, x\ge 0$. Although IP is NP-complete \cite{Karp72}, its instances are solved routinely in practice using cutting planes and branch-and-cut techniques \cite{NW14}. The relaxed problem of verifying integer feasibility of the system $Ax=b$ is equivalent to verifying whether $b$ lies in the lattice generated by the columns of $A$. Thus, the relaxation problem is the membership testing problem in a lattice. It is solvable efficiently and is a natural pre-processing step to solving IPs. Furthermore, if the number of constraints $n$ in the problem is very large, then it would be helpful to run a tester that reads only a partial set of coordinates of the input $b$ to verify if $b$ could lie in the lattice generated by the columns of $A$ or is far from it. 
If the test rejects, then this saves on the computational effort to search for a non-negative solution. \\


\noindent{\bf Cryptography.} 
In cryptographic applications, it is imperative to understand which lattices are difficult to test in order to ensure security of lattice-based cryptosystems. 
In some cryptanalytic attacks on lattice-based cryptosystems,  one needs to distinguish target vectors that are close to lattice vectors from those that are far from all lattice vectors, a problem commonly known as the gap version of the Closest Vector Problem (GapCVP).
An approach to address GapCVP is to use expensive distance estimation algorithms inspired by Aharonov and Regev \cite{AharonovR05} and Liu \etal\ \cite{LiuLM06}. Local testing of lattices is closely related to both  distance estimation \cite{PRR06} and GapCVP, and hence progress in the proposed testing model could lead to new insights in cryptanalytic attacks. \\


\noindent{\bf Complexity theory.} 
Lattices can be seen as coding theoretic objects naturally bringing features of error-correcting codes from the finite field domain to the real domain.
As such, a study of local testing (and correction) procedures for lattices naturally extends the classical notions of Locally Testable Codes (LTCs) and Locally Decodable Codes (LDCs),  which are in turn of significance to
computational complexity theory (for example in constructing probabilistically checkable proofs
and hardness amplification, among numerous other applications). 
Characterizing local testability, explicitly  initiated by Kaufman and Sudan \cite{KS08}, has been an intensely investigated direction in the study of LTCs. We believe that an analogous investigation of lattices is likely to bring new insights and new connections in property testing. \\


\noindent{\bf Lattice-based communication.} 
Lattices are a major technical tool in communication systems as the analogue of error-correcting codes over reals, for applications such as wireless communication and transmission over analog lines. 
In lattice-coding, the message $m$ is mapped to a point $c$ in a chosen lattice $L$.  
The codeword $c$ is transmitted over an analog channel. 
If the encoded message gets corrupted by the channel, then the channel output may not be a lattice point, thus enabling transmission error detection. In order to correct errors, computationally expensive decoding algorithms are employed. 
Instead, the receiver may perform a local test for membership in the lattice beforehand, allowing the costly decoding computation to run only when there is a reasonably high chance of correct decoding. \\

We now give an informal description of our testing model motivated by its application in lattice-coding. 
The transmission of each coordinate of a lattice-codeword over the analog channel consumes power that is proportional to the square of the transmitted value. Thus the power consumption for transmitting the lattice-codeword $c\in L\subset \R^n$ is proportional to its  squared $\ell_2$ norm. The power consumption for transmitting a codeword over the channel is usually constrained by a power budget. The noise vector is also subject to a bound on its power.
The power budget for transmission is typically formulated by considering the lattice-code $C(L)$ defined by the set of lattice points $c\in L$ that satisfy $\sum_{i=1}^n c_i^2 \leq \sigma n$ for some constant power budget $\sigma > 0$. 
In order to ensure that the receiver can tolerate adversarial noise budget $\delta$ per channel use, the shortest nonzero vector $v\in L$ should be such that $\sum_{i=1}^n v_i^2 \ge \delta n$. Thus, the \emph{relative distance} of the lattice-code $C(L)$ is defined to be $\sum_{i=1}^n v_i^2/n$, where $v\in L$ is a shortest nonzero lattice vector.  
The \emph{rate} of a lattice-code $C(L)$ is defined to be $(1/n)\log{|C(L)|}$ (note that this quantity could be larger than 1). In this work, an \emph{asymptotically good family of lattices},  is one that achieves rate and relative distance that are both lower bounded by a nonzero constant. Such families are ideal for use in noisy communication channels. 

We define a notion of a tester that will be useful as a pre-processor for decoding,  and is similar to the established notion of  a code tester: An $\ell_2$-tester of a lattice $L$ for a given distance parameter $\eps>0$ is a probabilistic procedure that given an input $t\in \R^n$, queries at most $q$ coordinates of $t$, accepts with probability at least $2/3$ if $t\in L$, and rejects with probability at least $2/3$ if $\sum_{i=1}^n(t_i - w_i)^2 \geq \eps n$ for every $w \in L$. 

We start by formalizing the model and stating two main motivating goals.


\subsection{Testing model}

In the above application, we focused on $\ell_2$ distances. We now formalize the notion of testing lattices for $\ell_p$ distances, which is the natural notion of distance for real-valued inputs.  
We remark that despite an extensive literature on property testing with respect to the Hamming distance, testing under $\ell_p$ distance was  only recently proposed for systematic investigation by  Berman \etal\ \cite{BRY14} in the context of testing non-algebraic properties.
The $\ell_p$ distance between $x,y\in \R^n$ is defined as $d_p(x, y):=\|x-y\|_p=(\sum_{i\in [n]} |x_i-y_i|^p)^{1/p}$.  The distance from $v\in \R^n$ to $L$ is $d_p(v, L):=\min_{u\in L} d_p(v, u).$  Denote 
the $\ell_p$ norm of the real vector $1^n$ by $\|1^n\|_p$. 
For a lattice $L$, we denote the subspace of the lattice by $span(L)$. 
We focus on integral lattices, which are sub-lattices of $\Z^n$, as these are the most commonly encountered lattices in applications\footnote{Arbitrary lattices can be approximated by rational lattices and rational lattices can be scaled to integral lattices.}. 

\begin{definition}[Local test for lattices]
\label{defn:local-test}
An $\ell_p$-tester $T(\eps,c,s,q)$ for a lattice $L\subseteq \Z^n$ is a probabilistic algorithm that queries $q$ coordinates of the input $t \in \R^n$, and 
\begin{itemize}
\item (completeness) {\em accepts}  with probability at least $1-c$ if  $t \in L$,
\item (soundness) {\em rejects} with probability at least $1-s$ if $d_p( t, L) \geq \epsilon \cdot \|1^n \|_p$
(we call such a vector $t$ to be $\eps$-far from $L$).
\end{itemize}
 If $T$ always accepts inputs $t$ that are in the lattice $L$ then it is called $1$-sided, otherwise it is $2$-sided. If the queries performed by $T$ depend on the answers to the previous queries, then $T$ is called adaptive, otherwise it is called non-adaptive.  
\end{definition}

A test $T(\eps,0,0, q)$ is a test with  perfect completeness and perfect soundness. {\em $1$-sided testers} (i.e., testers with perfect completeness) are useful as a pre-processing step, as mentioned earlier.

 An {\em asymptotically good} family of lattices $L(n)$ satisfies: 1) $\min_{v\in L(n)}\|v\|_p^p/n=\Omega(1)$, and 2) contains  $2^{\Omega(n)}$ lattice points in the origin-centered $\ell_p$-ball of radius $n^{1/p}$.

Similar to the application in lattice-coding and locally testable codes, a main question in $\ell_p$-testing of lattices is the following: 
\begin{question} \label{question-1}
Is there an asymptotically good family of lattices that can be tested for membership with constant number of queries? 
\end{question}

\noindent Motivated by the applications in IP and cryptography, we identify another fundamental question in $\ell_p$-testing of lattices: 
\begin{question} \label{question-2}
What properties of a given lattice enable the design of $\ell_p$-testers with constant query complexity?
\end{question}

\noindent{\bf Tolerant Testing.} Many applications can tolerate a small amount of noise in the input. Parnas \etal\ \cite{PRR06} introduced the notion of tolerant testing 
to account for a small amount of noise in the input. 
Tolerant testing has been studied in the context of codes (e.g. \cite{GuruswamiR05, KoppartyS09}),  and in the context of properties of real-valued data in the $\ell_p$ norm (e.g. \cite{BRY14}). We extend the tolerant testing model to lattices as follows.

\begin{definition}[Tolerant local test for lattices]
An $\ell_p$-tolerant-tester $T(\eps_1,\eps_2,c,s,q)$ for a lattice $L\subseteq \Z^n$ is a probabilistic algorithm that queries $q$ coordinates of the input $t \in \R^n$, and 
\begin{itemize}
\item (completeness) {\em accepts}  with probability at least $1-c$ if  $d_p(t,L)\le \eps_1 \cdot \|1^n \|_p$,
\item (soundness) {\em rejects} with probability at least $1-s$ if $d_p( t, L) \geq \eps_2 \cdot \|1^n \|_p$.
\end{itemize}
\end{definition}
Tolerant testing with parameter $\eps_1=0$ corresponds to the notion of testing given in Definition \ref{defn:local-test}. Tolerant testing and distance approximation are closely related notions. In fact, in the Hamming space, the ability to perform tolerant testing for \emph{every} choice of $\eps_1<\eps_2$ can be exploited to approximate distances (using a binary search) \cite{PRR06}.\\

\noindent{\bf Analogy with code testers.} A common notion of testing for membership in \emph{error-correcting codes} requires that inputs at \emph{Hamming distance} at least $\eps n$ from the code be rejected. 
(This notion is only relevant when the covering radius of the code is larger than $\eps n$.) We include the common definition here, and note that stronger versions of testing have also been considered in the literature \cite{Goldreich10,GuruswamiR05}. 

\begin{definition}[Local test for codes] \label{defn:code-test}
A tester $T(\eps,c,s,q)$ for an error-correcting code $C\subseteq \F^n$ is a probabilistic algorithm that makes $q$ queries to the input $t \in \F^n$, and 
\begin{itemize}
\item (completeness) {\em accepts}  with probability at least $1-c$ if  $t \in C$, and 
\item (soundness) {\em rejects} with probability at least $1-s$ if $d_H( t, C) \geq \epsilon \cdot n$, where $d_H(u,v):=|\{i\in [n]:u(i)\neq v(i)\}|$ denotes the Hamming distance between $u$ and $v$, and $d_H(t, C):=\min_{c\in C} d_H(t,c)$ (we call such a vector $t$ to
be $\eps$-far from $C$).
\end{itemize}  
\end{definition}


\subsection{Our contributions}

We initiate the study of membership testing in point lattices from the perspective of sublinear algorithms. Our contributions draw on connections between lattices and codes, and on well-known techniques in property testing. 

\subsubsection{Upper and lower bounds for testing specific lattice families}
We make progress towards Question \ref{question-1} 
by focusing on an asymptotically good family of sets constructed from linear codes, via the so-called ``code formula''  \cite{ForneyI}. 
We show upper and lower bounds on the query complexity of $\ell_1$-testers for code formulas, as a function of the query complexity of the constituent code testers. \\

\noindent {\bf Code formula lattices.} 
For simplicity, in what follows we will slightly abuse notation and use binary code $C\subseteq \{0,1\}^n$ to denote both the code viewed over the field $\F_2=\{0,1\}$ and the code embedded into $\R^n$ via the trivial embedding $0\mapsto 0$ and $1\mapsto 1$. All the arithmetic operations in the code formula refer to operations in  $\R^n$. For two sets $A$ and $B$ of vectors we define $A+B:=\{a+b\mid a\in A, b\in B\}$. 

\begin{definition}[Code Formula] 
Let $C_0 \subseteq C_1 \subseteq \cdots \subseteq C_{m-1} \subseteq C_m = \mathbb{F}_2^n$ be a family of nested binary linear codes. Then the code formula constructed from the family is defined as 
\[  C_0 + 2C_1 + \cdots + 2^{m-1}C_{m-1} + 2^m \mathbb{Z}^n. \]
Here, $m$ is the {\em height} of the code-formula.

If the family satisfies the \emph{Schur product condition}, namely, $c_1 * c_2 \in C_{i+1}$ for all codewords $c_1, c_2 \in C_i$, where the `*' operator is the coordinate-wise (Schur) product $c_1*c_2=\langle (c_1)_i \cdot (c_2)_i\rangle_{i\in[n]}$, then the code-formula forms a {\em lattice} (see \cite{KO14}) and we denote it by $L(\langle C_i \rangle_{i=0}^{m-1})$. 

\end{definition}

\noindent {\bf Significance of code formula lattices.} 
Code formula lattices with height one already have constant rate if the constituent code $C_0$ has minimum Hamming distance $\Omega(n)$. Unfortunately, these lattices have tiny relative minimum distance (since $2\Z^n$ has constant length vectors). 
However, code formulas of larger height achieve much better relative distance. In particular, it is easy to see that code formula lattices of height $m \geq \log{n}$ in which each of the constituent codes $C_i$ has minimum Hamming distance $\Omega(n)$ give asymptotically good families of lattices \cite{GZ06, conway99}. The code formula lattice constructed from a family of codes that satisfies the Schur-product condition is equivalent to the lattice constructed from the same family of codes by Construction D \cite{LS71, conway99, KO14}. Construction-D lattices are primarily used in communication settings, e.g. see Forney \cite{ForneyI}. 

In this work we design a tester for code formula lattices using testers for the constituent codes. 

\begin{restatable}{theorem}{thmConstrDTest}
\label{thm:constrD:test}
Let $0<\eps,s<1$ and $C_0\subseteq C_1 \subseteq \cdots \subseteq C_{m-1}\subseteq \{0,1\}^n$ be a family of binary linear codes satisfying the Schur product condition.
Suppose every $C_i$ has a 1-sided tester 
$T_i(\eps/m2^{i+1}, 0, s, q_i)$. 
Then, there exists an $\ell_1$-tester $T(\eps, 0, s , q)$ for 
the lattice $L(\langle C_i \rangle_{i=0}^{m-1})$
with query complexity  
$$
q =O\left(\frac{1}{\eps}\log{\frac 1s}\right)+\sum_{i=1}^{m-1} q_i.
$$
\end{restatable}

Next, we show a lower bound on the query complexity for testing membership in code formula lattices, using lower bounds for testing membership in the constituent codes.

\begin{restatable}{theorem}{thmLowerBoundCodeTester}
\label{theorem:lower-bound-tester-for-code}
Let $0<\eps,c,s<1$ and $C_0\subseteq C_1 \subseteq \cdots \subseteq C_{m-1}\subseteq \{0,1\}^n$  be a family of binary linear codes satisfying the Schur product condition. 
Let $q_i=q_i(\eps, c,s)$ be such that any (possibly adaptive, 2-sided) $\ell_1$-tester $T_i(\eps,c,s,q')$ for $C_i$  satisfies $q'=\Omega(q_i)$, for every $i=0,1,\ldots, m-1$.
Then every (possibly adaptive, 2-sided) $\ell_1$-tester $T(\epsilon,c,s,q)$ for the lattice $L(\langle C_i \rangle_{i=0}^{m-1})$ has query complexity 
\[
q= \Omega\left(\max\left\{\frac{1}{\eps}\log{\frac{1}{s}},\max_{i=0,1,\ldots,m-1}q_i\right\}\right).
\] 
\end{restatable}

\noindent {\em Code formula lattices from Reed-Muller codes.} We instantiate the upper and lower bounds on the query complexity for a common family of code formula lattices constructed using Reed-Muller codes \cite{ForneyI} to obtain nearly matching upper and lower bounds. We recall Reed-Muller codes below.

\begin{definition}[Reed Muller Codes]
Each codeword of a binary Reed-Muller code $RM(k,r)\subseteq \F_2^{2^r}$ corresponds to a polynomial $p(x)\in \F_2[x]$ in $r$ variables of degree at most $k$ evaluated at all $2^{r}$ possible inputs $x\in \F_2^r$. 
\end{definition}
For the family of Reed-Muller codes in $\F_2^{2^r}$, it is well-known that  $RM(0, r) \subseteq 
RM(1, r) \subseteq RM(2, r) \subseteq RM(3, r) \subseteq \cdots \subseteq RM(r-1, r)\subseteq RM(r, r)=\F_2^{2^r}$. 
A particular family of RM codes that leads to code formula lattices is $\langle RM(k_i, r) \rangle_{i = 0}^{\log r}$, with $k_i=2^i$. 
Indeed, it can be easily verified that this family satisfies the Schur product condition since Reed-Muller codewords are evaluation tables of multivariate polynomials over the binary field and product of two degree $k$ polynomials is a degree $2k$ polynomial. Hence for height $m \leq \log r$ the construction $\langle RM(2^i, r) \rangle_{i = 0}^{m-1}$ gives rise to a lattice. 

\begin{restatable}{cor}{coroConstrDTest}
\label{coro:constrD:test}
Let $0 \leq k_0 < k_1 < \cdots < k_{m-1} < r $ be integers such that the family of Reed-Muller codes $RM(k_0,r)\subseteq RM(k_1,r)\subseteq\cdots \subseteq RM(k_{m-1},r)$ satisfies the Schur product condition. Let $0<\eps,s<1$ and $L$ be the lattice obtained from this family of codes using the code formula construction:
\[ L = RM(k_0,r)+ 2 RM(k_1, r)+ \cdots + 2^{m-1} RM(k_{m-1}, r) + 2^m \Z^{2^r}. \]
Then, there exists an 
$\ell_1$-tester $T(\eps, 0, s, q)$ for $L$ with query complexity $$q(\eps, s) =
O\left(2^{k_{m-1}} \cdot \frac{1}{\eps}  \log {\frac 1s} \right).$$
\end{restatable}
In particular, when the height $m$ and the degrees  are constant, the query complexity of the tester is a constant.

For the lower bound, we obtain the following corollary using known lower bounds for testing Reed-Muller codes.
\begin{restatable}{cor}{CoroConstrDLB}
\label{coro:constrD:lb}
Let $0 \leq k_0 < k_1 < \cdots < k_{m-1} < r $ be integers such that the family of Reed-Muller codes $RM(k_0,r)\subseteq RM(k_1,r)\subseteq\cdots \subseteq RM(k_{m-1},r)$ satisfies the Schur product condition.
Let $0<\eps,c,s<1$ be constants and $L$ be the lattice obtained from this family of codes using the code formula construction:
\[ L = RM(k_0,r)+ 2 RM(k_1, r)+ \cdots + 2^{m-1} RM(k_{m-1}, r) + 2^m \Z^{2^r}. \]
Then, every (possibly $2$-sided, adaptive) $\ell_1$-tester $T(\eps,c,s,q)$ for $L$ has query complexity $$q=\Omega(2^{k_{m-1}}).$$
\end{restatable}

We note that for code formula lattices obtained from Reed-Muller codes, Corollaries \ref{coro:constrD:test} and \ref{coro:constrD:lb} show matching bounds (up to a constant factor depending on $\eps, s$). \\\\

\noindent{\bf Random lattices.} We also observe that random lattices obtained from binary, random LDPC codes are not testable with a small number of queries. 
Indeed, consider the following distribution of random lattices (e.g., \cite{ELS05,BHR05}):
 For constants $b < a$, let $m=nb/a$ and let $H\in \F_2^{m\times n}$ be a random matrix such that each row and column has exactly $a$ and $b$ non-zeroes respectively. 
Consider the linear code $C_{a,b}:=\{x\in \F_2^n: Hx=0 (\text{mod}\ 2)\}$ and the code formula lattice $L(C_{a,b})$ associated with the linear code $C_{a,b}$.

\begin{theorem}
There exist constants $a$, $b$, $\epsilon$, $c$, $s$ such that every (possibly 2-sided, adaptive) $\ell_1$-tester $T(\epsilon, c, s,q)$ for $L(C_{a,b})$ has query complexity $q=\Omega(n)$. 
\end{theorem}

The above theorem follows as an immediate corollary of Theorem~\ref{theorem:lower-bound-tester-for-code} and of Theorem 3.7 in  \cite{BHR05}. 

\subsubsection{Tolerant testing code formulas}
We also obtain upper bounds for 
tolerantly testing code formula lattices.
 
\begin{restatable}{theorem}{thmTolerantTester}
\label{theorem:tolerant-tester-code-formula}
Let $0 < \eps_1, \eps_2, c, s  < 1 $ and $C_0\subseteq C_1\subseteq \cdots\subseteq C_{m-1} \subseteq \{0,1\}^n$ be a family of binary linear codes satisfying the Schur product condition. Suppose every $C_i$ has a tolerant tester $T_i( 2 \eps_1, \frac{\eps_2 }{m2^{i+1}}, \frac{c}{m+1}, s, q_i)$. Let $\gamma = \min \{ c/(m+1), s \}$, $\eps_2 > m2^{m+1} \eps_1$. Then there exists an $\ell_1$-tolerant-tester $T(\eps_1, \eps_2, c, s, q)$ for the lattice $L( \langle C_i \rangle_{i=0}^{m-1})$ with query complexity 
\[ 
q = O\left(\frac{1}{(\eps_2- 2\eps_1)^2} \log\left(\frac{1}{\gamma}\right) \right) + \sum_{i=0}^{m-1} q_i . 
\]
\end{restatable}

\begin{restatable}{cor}{CoroTolerantRM}
\label{coro:CoroTolerantRM}
Let $0 \leq k_0 < k_1 < \cdots < k_{m-1} < r $ be integers such that the family of Reed-Muller codes $RM(k_0,r)\subseteq RM(k_1,r)\subseteq\cdots \subseteq RM(k_{m-1},r)$ satisfies the Schur product condition.
Let $L$ be the lattice obtained from this family of codes using the code formula construction:
\[ L = RM(k_0,r)+ 2 RM(k_1, r)+ \cdots + 2^{m-1} RM(k_{m-1}, r) + 2^m \Z^{2^r}. \]
Then there exists a $\ell_1$-tolerant-tester $T(\eps_1, \eps_2, 1/3,1/3, q)$ for $L$ for all  $\eps_1\leq \frac{c_1'}{2^{k_{m-1}}} $,  $\eps_2\geq \frac{c_2' m }{2^{k_0-1}}$ (for some constants $c'_1$ and $c'_2$) with query complexity $q=O(2^{k_{m-1}}\cdot \log m)$.
\end{restatable}


\subsubsection{A canonical/linear test for lattices} 
Our next result makes progress towards addressing Question \ref{question-2}. We show 
a reduction from any given arbitrary test to a \emph{canonical linear test}, thus suggesting that it is sufficient to design \emph{canonical linear tests}  for achieving low query complexity. 
In order to describe the intuition behind a canonical linear test, we first illustrate how to solve the membership testing problem when all coordinates of the input are known. For a given lattice $L$, its {\em  dual lattice} is defined as $$L^{\bot}:=\{ u\in span(L)\mid \langle u, v\rangle \in \Z, \text{ for all } v\in L\}.$$ It is easy to verify that $(L^\perp)^{\perp}=L$. Furthermore, a vector $v\in L$ if and only if for all $u \in L^{\bot}$, we have $\langle u, v \rangle \in \mathbb{Z}$. Thus, to test membership of $t$ in $L$ in the classical decision sense, it is sufficient to verify whether $t$ has integer inner products with a set of basis vectors of the dual lattice $L^{\bot}$. 
Inspired by this observation, 
 we define a canonical {\em linear test} for  lattices as follows. For a lattice $L\subseteq \R^n$ and $J\subseteq [n]$, let $L^{\bot}_J  := \{ x \in L^{\bot} \mid supp(x) \subseteq J \}$, where $supp(x)$ is the set of non-zero indices of the vector $x$.

\begin{definition}[Linear Tester]
A {\em linear tester} for a lattice $L\subseteq \Z^n$ is a probabilistic algorithm which queries a subset $J = \{ j_1, \ldots, j_q \} \subseteq [n]$ of coordinates of the input $t \in \R^n$ and accepts $t$ if and only if 
$\langle t, x \rangle \in \mathbb{Z}$ for all $x \in L^{\bot}_J$.
\footnote{
Verifying whether $\langle t, x \rangle\in \Z$ for all $x\in L^{\bot}_J$ can be performed efficiently by checking inner products with a set of  basis vectors of the lattice $L^{\bot}_J$.}
\end{definition}

\noindent{\bf Remark.} By definition, the probabilistic choices of a linear tester are only over the set of coordinates to be queried: upon fixing the coordinate queries, the choice of the algorithm to accept or reject is fully determined. 
Furthermore, a linear tester is $1$-sided since if the input $t$ is a lattice vector, then for every dual vector $u\in L^{\bot}$, the inner product $\inprod{u}{t}$ is integral, and so it will be accepted with probability $1$. \\

We show that non-adaptive linear tests are nearly as powerful as 2-sided adaptive tests for a full-rank lattice. We reduce any (possibly 2-sided,  and adaptive) test for a  full-rank lattice to a non-adaptive linear test for the same distance parameter $\eps$, with a small increase in the query complexity and the soundness error. 

\begin{theorem}\label{theorem:2-sided-to-1-sided-non-adaptive-linear}
Let $L\subseteq \Z^n$ be a lattice with rank$(L)=n$. If there exists an adaptive $2$-sided $\ell_p$-tester $T(\eps,c,s,q)$ with query complexity $q=q_T(\eps,c,s)$, then there exists a non-adaptive linear $\ell_p$-tester $T'(\eps,0,c+s,q')$ with query complexity $q'=q_T(\eps/2,c,s)+O((1/\eps^p)\log{(1/s)})$. 
\end{theorem}

Furthermore, if we are guaranteed that the inputs are in $\Z^n$, then the query complexity of the test $T'$ above can be improved to be identical to that of $T$ (up to a constant factor in the $\eps$ parameter). The increase in the query complexity comes from an extra step used to verify the integrality of the input.  

Theorem \ref{theorem:2-sided-to-1-sided-non-adaptive-linear} suggests that, for the purposes of designing a tester with small query complexity, it is sufficient to design a non-adaptive linear tester, i.e., it suffices to only identify the probability distribution for the coordinates that are queried. Moreover, this theorem makes progress towards Question \ref{question-2},  since it shows that 
a lower bound on the query complexity of non-adaptive linear tests for a particular lattice implies a lower bound on the query complexity of all tests for that lattice. Thus in order to understand the existence of low query complexity tester for a particular lattice, it is sufficient to examine the existence of low query complexity \emph{non-adaptive linear} tester for that lattice. 

We note that Theorem \ref{theorem:2-sided-to-1-sided-non-adaptive-linear} is the analogue of the result of \cite{BHR05} for linear error-correcting codes. In section \ref{subsection:proof-overviews}, we comment on the comparison between our proof and that in \cite{BHR05}. 


\subsubsection{Testing membership of inputs outside the span of the lattice}
We also observe a stark difference between the membership testing problem for a linear code, and the membership testing problem for a lattice. In the membership testing problem for a linear code $C\subseteq \F^n$ defined over a finite field that is specified by a basis, the input is assumed to be a vector in $\F^n$ and the goal is to verify whether the input lies in the span of the basis (see definition \ref{defn:code-test}).
As opposed to codes, for a lattice $L\subseteq \R^n$, the input is an arbitrary real vector, and the goal is to verify whether the input is a member of $L$, and  not to verify whether the input is a member of the span of the lattice. 
Thus, the inputs to the lattice membership testing problem could lie either in $span(L)$,  or outside $span(L)$. Interestingly, for some lattices it is easy to show strong lower bounds on the query complexity if the inputs are allowed to lie outside $span(L)$, thus suggesting that such inputs are hard to test. 

\begin{restatable}{theorem}{thmLowerBoundForOutsideSpan}
\label{thm:lower-bound-for-outside-span}
Let $L\subseteq \Z^n$ be a lattice of rank $k$. Let $P\subseteq[n]$ be the support of the vectors in $span(L)^{\bot}$. 
Let $0< \eps, c, s < 1$. Every non-adaptive $\ell_p$-tester $T(\eps, c, s, q)$ for $L$ for inputs in $\R^n$ has query complexity 
$$q=\Omega(\lvert P \rvert).$$
\end{restatable}

On the other hand, testers for inputs in the $span(L)$ can be lifted to obtain testers for all inputs (including inputs that could possibly lie outside $span(L)$). 

\begin{restatable}{theorem}{thmUpperBoundForOutsideSpan}
\label{thm:upper-bound-for-outside-span}
Let $L\subseteq \Z^n$ be a lattice of rank $k$. Let $P\subseteq[n]$ be the support of the vectors in $span(L)^{\bot}$. 
Let $0 < \eps, c, s < 1$, and suppose $L$ has an $\ell_p$-tester $T(\eps, c, s, q)$ for inputs $t\in span(L)$. Then $L$ has a tester $T'(2\eps, c, s, q')$ for inputs in $\R^n$ with query complexity
$$q'\leq q+\lvert P \rvert.$$
\end{restatable}

Theorem \ref{thm:upper-bound-for-outside-span} implies that for lattices $L$ of rank at most $n-1$, if the membership testing problem for inputs that lie in $span(L)$ is solvable using a small number of queries and if $span(L)^{\bot}$ is supported on few coordinates, then the membership testing problem for all inputs (including those that do not lie in $span(L)$) is solvable using a small number of queries. \\\\
\noindent {\bf Knapsack Lattices.} Theorem \ref{thm:lower-bound-for-outside-span} implies a linear lower bound for non-adaptively testing a well-known family of lattices, known as {\em knapsack lattices},  which have been investigated in the quest towards lattice-based cryptosystems \cite{merkle78, shamir83, odlyzko90}.
We recall that a knapsack lattice is generated by a set of basis vectors $B=\{b_1,\ldots, b_{n-1}\}, b_i\in \R^{n}$ that are of the form
\begin{align*}
b_1 &=(1,0,\ldots, 0,a_1)\\
b_2 &=(0,1,\ldots, 0,a_2)\\
&\vdots\\
b_{n-1} &=(0,0,\ldots, 1,a_{n-1})
\end{align*}
where $a_1,\ldots, a_n$ are integers. We denote such a knapsack lattice by $L_{a_1,\ldots, a_{n-1}}$.

\begin{restatable}{cor}{CorKnapsackLowerBoundOutsideSpan}
\label{cor:knapsack-lower-bound-outside-span}
Let $a_1,\ldots,a_n$ be integers and $0< \eps, c,s < 1$. Every non-adaptive $\ell_p$-tester $T(\eps,c,s,q)$ for $L_{a_1,\ldots, a_n}$ has query complexity
$$ q=\Omega(n).$$ 
\end{restatable}

However, knapsack lattices with bounded coefficients are testable with a constant number of queries if the inputs are promised to lie in $span(L)$.

\begin{restatable}{theorem}{ThmKnapsackTester}
\label{thm:knapsack-tester}
Let $a_1,\ldots,a_n$ be integers with $M=\max_{i\in [n]}|a_i|^p$ and $0< \eps,s < 1$. 
There exists a non-adaptive $\ell_p$-tester $T(\eps,0,s,q)$ for $L_{a_1,\ldots, a_n}$ with query complexity 
$
q = O\left( \frac{M}{\eps^p}\cdot \log\frac{1}{s}\right)
$, if the inputs are guaranteed to lie in $span(L)$.
\end{restatable}

Theorem \ref{thm:knapsack-tester} indicates that the large lower bound suggested by Theorem \ref{thm:lower-bound-for-outside-span} could be circumvented for certain lattices if we are promised that the inputs lie in $span(L)$. The assumption that the input lies in $span(L)$ is natural in decoding problems for lattices.


\section{Overview of the proofs}
\label{subsection:proof-overviews}

\subsection{Upper and lower bounds for testing general code formula lattices}

The constructions of a tester for Theorem \ref{thm:constrD:test} and a tolerant tester for Theorem \ref{theorem:tolerant-tester-code-formula}  
follow the natural intuition that in order to test the lattice one can test the underlying codes individually. The proof relies on a triangle inequality that can be derived for such lattices.
The application to code-formula lattices constructed from Reed-Muller codes follows from the tight analysis of Reed-Muller code testing from \cite{BKSSZ10}, which guarantees constant rejection probability of inputs that are at distance proportional to the minimum distance of the code.

While the tester that we construct from code testers for the purposes of proving Theorem \ref{thm:constrD:test} is an adaptive linear test, there is a simple variant that is a non-adaptive linear test with at least as good correctness and soundness. 
(see Remark \ref{remark:linearD} for a formal description).

The lower bound (Theorem \ref{theorem:lower-bound-tester-for-code}) relies on the  fact that if an input $t$ is far from the code $C_k$ in the code formula construction, then the vector $2^k t$ is far from the lattice  (Lemma \ref{lemma:distance-to-lattice}). Moreover, if $t\in C_k$ then $2^k t$ belongs to the lattice. Therefore a test for the lattice can be turned into a test for the constituent codes. \\


\subsection{From general tests to canonical tests}

We briefly outline our reduction for Theorem \ref{theorem:2-sided-to-1-sided-non-adaptive-linear}.  Suppose $T(\epsilon, c, s, q)$ is a 2-sided, adaptive tester with query complexity  $q=q_T(\eps,c,s)$ for a full rank integral lattice $L$. Such a tester handles all real-valued inputs. We first restrict $T$ to a test that processes only integral inputs in the bounded set ${\cal Z}_d=\{0, 1, \ldots, d-1\}$ (for some carefully chosen $d$),  and so the restricted test inherits all the parameters of $T$. We remark that ${\cal Z}_d\subset \Z$ is a subset of integers, and it should not be confused with $\Z_d$, the ring of integers modulo $d$. 

A  key ingredient in our reduction is choosing the appropriate value of $d$ in order to enable the same guarantees as that of codes.
We choose $d$ such that $d\Z^n\subseteq L$. Such a $d$ always exists \cite{MiccLN12}. This choice of $d$ allows us to add any vector in $V=L\mmod d$ (embedded in $\R^n$) to any vector $x\in \R^n$ without changing the distance of $x$ to $L$ in any $\ell_p$-norm (see Proposition \ref{prop:mod-preserves-distance}).

Since our inputs are now integral and bounded, any adaptive test can be viewed as a distribution over deterministic tests, which themselves can be viewed as decision trees. 
This allows us to proceed along the same lines as in the reduction for codes over finite fields of \cite{BHR05}. 

We exploit the property  that adding any vector in $V$ to any vector $x\in \R^n$ does not change the distance to $L$. In the first step of our reduction we add a random vector in $V$ to the input and perform a  probabilistic {\em  linear} test. The idea  is that one can relabel the decision tree of any test according to the decision tree of a linear test, such that the error shifts from the positive (yes) instances to the negative (no) instances (see Lemma \ref{lemma:2-sided-to-1-sided-linear}). A simple property of lattices used in this reduction is that if the set of queries $I$ and answers $a_I$ do not have a local witness for non-membership in the lattice (in the form of a dual lattice vector $v$ supported on $I$ such that $\langle w_I, v_I\rangle\not \in \Z$), then there exists  $w\in L$ that extends $a_I$ to the remaining set of coordinates (i.e., $a_I=w_I$). 

In the next step we remove the adaptive aspect of the test to obtain a non-adaptive linear test for inputs in ${\cal Z}_d^n$ (see Lemma \ref{lemma:adaptive-to-non-adaptive}). We obtain this tester by performing the adaptive queries on a randomly chosen vector in $V$ (and not on the input itself) and rejecting/accepting according to whether there exists a local witness for the non-membership of the input queried on the same coordinates.

We then  lift this test  to a non-adaptive linear test for inputs in $\Z^n$, by simulating the test over ${\cal Z}_d^n$ on the same queried coordinates but using the answers obtained after taking modulo $d$. 
Owing to the choice of $d$, this does not change the distance of the input to the lattice (see Lemma \ref{lemma:bounded-targets}). 
 
Finally, we extend this test to a non-adaptive linear test for inputs in $\R^n$ by performing some additional queries to rule out inputs that are not in $\Z^n$. For this, we design a tester for the integer lattice $\Z^n$ with query complexity $O((1/\eps^p)\log{(1/s)})$. 
This final step of testing integrality increases the overall query complexity to $q_T(\eps/2,c,s)+O((1/\eps^p)\log{(1/s)})$ (see Lemma \ref{lemma:integer-targets}).\\


\noindent{\bf Organization.}  
We prove the upper bound and lower bound for testing code formula constructions (Theorem~\ref{thm:constrD:test}, Theorem~\ref{theorem:lower-bound-tester-for-code}) and its instantiations to Reed-Muller codes (Corollary~\ref{coro:constrD:test}, Corollary~\ref{coro:constrD:lb} ) in Section~\ref{sec:testingconstrD}.
The upper bound for tolerant testing code-formula constructions (Theorem \ref{theorem:tolerant-tester-code-formula}) and its instantiations to Reed-Muller codes (Corollary~\ref{coro:CoroTolerantRM}) are proved in Section~\ref{sec:toltestingconstrD}.  
We present the formal lemmas needed to prove Theorem \ref{theorem:2-sided-to-1-sided-non-adaptive-linear}  and their proofs in Section \ref{sec:reduction}. 
We address non-full-rank lattices and prove Theorems \ref{thm:lower-bound-for-outside-span}, \ref{thm:upper-bound-for-outside-span}, Corollary \ref{cor:knapsack-lower-bound-outside-span} and Theorem~\ref{thm:knapsack-tester} in Section~\ref{sec:inputs-outside-span}.
 
\section{Testing Code-Formula Lattices}
\label{sec:testingconstrD}
\newcommand{\Frac}{\mathsf{frac}}

\subsection{Upper Bounds for Code-Formula Lattices}
\label{subsec:upper-bound-code-formula}
In this section we construct a tester for testing membership in lattices obtained from the code formula construction using a tester for the constituent codes.

\thmConstrDTest*

\begin{proof}
Let $L := L(\langle C_i \rangle_{i=0}^{m-1})$.
First, we use Lemma~\ref{lemma:integer-targets} to reduce the task to testing integral inputs for
distance parameter $\eps/2$.
According to the lemma, it suffices to show that all such inputs can be tested with $1$-sided error and using
$\sum_{i=1}^{m-1} q_i$ 
queries.

Now, let $w\in \Z^n$ denote the input.  
We may assume that all coordinates of $w$ are non-negative integers less than $2^m$. 
Otherwise, we can shift each coordinate by an appropriate (possibly different) multiple
of $2^m$ to make sure this condition holds. Observe that each such operation would correspond to shifting $w$ by an integer multiple of the lattice vector $2^m e_i$, where $e_i$ is the $i^{th}$ basis vector, and that translating a vector by a lattice point does not affect its distance to the lattice. 
Moreover, observe that this transformation can be applied implicitly, locally, and efficiently by the testing algorithm as the queries are made.

Let $w_0,\ldots, w_{m-1}\in \{0,1\}^n$ where $w_i(j)$ is the $(i+1)^{th}$ least significant bit in the binary decomposition of the $j^{th}$ coordinate of $w$. Thus we have $w=\sum_{i=0}^{m-1} 2^i w_i$. Once again, the coordinates of $w_i$ can be computed implicitly, locally and efficiently by the algorithm as the queries are made.

The tester $T$ would now proceed as follows: Run $T_i(\eps/(m2^{i+1}), 0, s, q_i)$ on $w_i$ for every $i=0,1,\ldots, m-1$. 
Accept if and only if all tests accept.

The overall query complexity of this tester is 
$\sum_{i=1}^{m-1} q_i$. 

The completeness of this tester is easy to deduce. Indeed, if $w \in L(\langle C_i \rangle_{i=0}^{m-1})$, then 
by definition of the code formula construction, there exist $\tilde{w_i}\in C_i$ for every $i=0,1,\ldots,m-1$ and an integer $\tilde{w}_m\in 2^m \Z^n$ such that $w=\sum_{i=0}^{m-1}2^i \tilde{w}_i + \tilde{w}_m$. Since entries of $w$ are non-negative integers
less than $2^m$, we must have $\tilde{w}_m = 0$. 
Moreover, since $\tilde{w_i} \in \{0,1\}^n$ and the binary representation is unique,
it must be that $\tilde{w}_i \in C_i$ for $i=0, \ldots, m-1$. 
That is,  each of the $w_i $ embedded in $\F_2^n$ are in $C_i$ and therefore, each $T_i(\eps/(2^{i+1} m), 0, s, q_i)$ (and thus the overall tester) will accept $w_i$. 

Before analyzing the soundness, we observe the following simple inequality. 

\begin{claim}
\label{claim:triangle-inequality-decomposition} 
$
d_1(w, L)\le d_1(w_0, C_0) + 2 d_1(w_1, C_1) + \cdots + 2^{m-1} d_1(w_{m-1}, C_{m-1}).  $
\end{claim}

\begin{proof}
Let $c_i \in \{0, 1\}^n$ be the closest codeword to $w_i$ in $C_i$ for every $i =0,1,\ldots,m-1$. From the definition of the code formula construction, we know that the vector $v = c_0 + 2c_1 + \cdots + 2^{m-1}c_{m-1}$ is a lattice vector. Therefore, 
\begin{align*}
\sum_{i=0}^{m-1} 2^i d_1(w_i,C_i) 
&= \sum_{i=0}^{m-1} 2^i \|w_i - c_i \|_1 \\
&\ge  \left \| \sum_{i=0}^{m-1} 2^i(w_i - c_i) \right \|_1 \quad \quad \text{(by the triangle inequality)}\\
&=d_1(w, {v}) \\
&\ge d_1(w,L).
\end{align*}
\end{proof}

Now, if $d_1(w,L)\ge \eps n/2$, then by Claim~\ref{claim:triangle-inequality-decomposition}
we have
\begin{align*}
d_1(w_0, C_0) + 2 d_1(w_1, C_1) + \cdots + 2^{m-1} d_1(w_{m-1}, C_{m-1}) 
\ge \eps n/2. 
\end{align*}
Therefore, by averaging,   
for some $i\in\{0,1,\ldots,m-1\}$ we must have $d_1(w_i,C_i)\ge (\eps / (m 2^{i+1}))n $.  
Thus the tester $T_i(\eps/(m2^{i+1}), 0, s, q_i)$ will reject with probability at least $1-s$. Hence the soundness follows.
\end{proof}

We now apply the result of Theorem \ref{thm:constrD:test} to the lattice obtained by applying code formula on a nested family of Reed-Muller codes.
In order to do so, we use the following result.

\begin{theorem} \label{theorem:RM-codes-upper-bound}\cite{BKSSZ10} 
For any $0 \leq k \leq r$ and $0 < s < 1$, $RM(k, r)$ has a 1-sided tester $T(\eps, 0, s, q(\eps, s))$ with query complexity $q(\eps, s) = O( (\log \frac1s)(2^{k} + \frac1\eps) )$ whose queries are each uniformly distributed.
\end{theorem}

Using the above result in Theorem~\ref{thm:constrD:test}, we obtain Corollary \ref{coro:constrD:test} whose proof we present next.

\coroConstrDTest*

\begin{proof}
From Theorem~\ref{thm:constrD:test}, for any $\eps>0$, there is a $T(\eps, 0,s, q(\eps,s))$ tester for  the code formula lattice $L(\langle C_i \rangle_{i=0}^{m-1})$  with  query complexity  $q(\eps,s)=O\left(\frac{1}{\eps}\log{\frac 1s}\right)+\sum_{i=0}^{m-1} q_i \left(\frac{\eps}{m 2^{i+1}},s\right)$, where $q_i(\eps, s)$ is the query complexity of testing the code $C_i$ (with distance parameter $\eps$ and soundness error $s$). 
By Theorem \ref{theorem:RM-codes-upper-bound},  each $RM(k_i, r)$ has a 1-sided tester $T_i(\eps_i, 0, s, q_i(\eps_i, s))$ with query complexity $q_i(\eps_i, s) = O(\log (1/s))(2^{k_i} +1/\eps_i)$. 
Therefore, the 1-sided tester $T(\eps, 0, s, q)$ for the code formula lattice $L(\langle RM(i, r) \rangle_{i=k_0}^{k_{m-1}})$ has query complexity 
\begin{eqnarray*}
q(\eps,s )&=&O\left(\frac{1}{\eps}\log{\frac 1s}\right)+\sum_{i=0}^{m-1} q_i \left(\frac{\eps}{m 2^{i+1}},s\right)\\
&=&O\left(\frac{1}{\eps}\log \frac 1s\right)\sum_{i=0}^{m-1}\left( 2^{k_i}+m 2^{i+1} \right)\\
&=&O\left(\frac{1}{\eps}\log \frac 1s \right)\left ( m  2^{m}+\sum_{i=0}^{m-1}2^{k_i}\right)\\
&=&O\left(\frac{1}{\eps}\log \frac 1s \right)\left ( m  2^{m}+2^{k_{m-1}}\right)\\
&=& O\left(\frac{1}{\eps}\log \frac 1s \cdot 2^{k_{m-1}}  \right).
\end{eqnarray*}

To see the last step of the above equation, recall that in order for $L(\langle RM(i, r) \rangle_{i=k_0}^{k_{m-1}})$ to be a lattice, we must have $k_i\geq 2^{i-1}$ for $i>0$, and in particular $k_{m-1}\geq 2^{m-2}$. So $2^{k_{m-1}}\geq 2^{2^{m-2}}\geq m 2^m$, for $m\geq 5$.
\end{proof}

\subsection{Lower bounds for Code-Formula Lattices}\label{sec:lower-bounds}

In this section, we  prove Theorem~\ref{theorem:lower-bound-tester-for-code}. 
We will use the following lemma. 

\begin{lemma}\label{lemma:distance-to-lattice}
Let $C_0,C_1,\ldots, C_{m-1}$ be a family of codes satisfying the Schur product condition and $L=L(\langle C_i \rangle_{i=0}^{m-1})$. Let $t\in \{0,1\}^n$ and $k\in \{0,1,\ldots,m-1\}$. Then
\[
d_1(t,C_k)\leq d_1(2^kt,L)\leq 2^k d_1(t,C_k).
\]
\end{lemma}

\begin{proof}
Since $2^k C_k$ is contained in $L$, $d_1(2^kt, L) \leq d_1(2^kt, 2^kC_k) = 2^k d_1(t, C_k)$. 
So, the distance of $2^kt$ to the lattice is at most $2^k d_1(t,C_k)$. We now show the inequality
\[
d_1(2^kt, L)\ge d_1(t,C_k).
\]
Let $v=\sum_{j=0}^{m-1}2^j c_j + 2^m z$ for some arbitrary $c_j\in C_j$  (for every $j \in \{0,1,\ldots,m-1\}$) and some $z\in \Z^n$, be an arbitrary lattice vector. We will show that $d_1(2^k t,v)\ge d_1(t,C_k)$. Let $u=c_k-t$, and $S \subseteq [n]$ be the support of $u$, then $|S|\ge d_1(t,C_k)$.

By Claim \ref{claim:all-coordinates-are-large}, $d_1(2^kt,v) \ge \sum_{i \in S} |v(i) - 2^k t(i)| \ge |S|$ (where $v(i)$ and $t(i)$ represent the $i$th entry in the respective vectors), which completes the proof.
\end{proof}

\begin{claim}\label{claim:all-coordinates-are-large}
For every $i\in S$, $|v(i) - 2^k t(i)|\ge 1$.
\end{claim} 

\begin{proof}
Let $i\in S$. Since $2^k u =2^k c_k - 2^k t$, we have  $2^k|u(i)|=2^k$. We also have
\[
|v(i) - 2^k t(i)|=\left |\sum_{j=0}^{k-1} 2^jc_{j}(i) + 2^k u(i)+\sum_{j=k+1}^{m-1}2^j c_{j}(i) + 2^{m}z(i) \right|.
\]
Since $c_{j}(i)\in \{0,1\}$ for every $j\in \{0,1,\ldots, k-1\}$, the first term in the above sum is at least zero and at most $2^k-1$. The maximum is achieved when all $c_j(i) = 1$ for all $j\in [k-1]$, and $u(i) = 1$. Hence, 
$
1\le \left |\sum_{j=0}^{k-1} 2^jc_{j}(i) + 2^k u(i)\right |\le 2^{k+1}-1.
$
Since $c_{k+1}(i)\in \{0,1\}$, we have
\[
1\le \left |\sum_{j=0}^{k-1} 2^jc_{j}(i) + 2^k u(i)+2^{k+1} c_{k+1}(i)\right |\le \left |\sum_{j=0}^{k-1} 2^jc_{j}(i) + 2^k u(i)\right |+\left| 2^{k+1} c_{k+1}(i)\right |\le 2^{k+2}-1.
\]
Proceeding similarly, since $c_{j}(i)\in \{0,1\}$ for $j=k+2,\ldots, m-1$, we have
\[
1\le \left |\sum_{j=0}^{k-1} 2^jc_{j}(i) + 2^k u(i)+\sum_{j=k+1}^{m-1}2^j c_{j}(i) \right|\le 2^{m}-1.
\]
Since $z_m\in \Z$, we conclude that 
$
\left |\sum_{j=0}^{k-1} 2^jc_{j}(i) + 2^k u(i)+\sum_{j=k+1}^{m-1}2^j c_{j}(i) + 2^{m}z(i) \right|\ge 1.
$
\end{proof}

\thmLowerBoundCodeTester*
\begin{proof}
Let $T(\epsilon, c, s, q)$ be a test for the code lattice, and let  $k\in \{0,1,\ldots,m-1\}$. We construct a tester $T_k(\eps,c,s,q)$ for $C_k$ as follows: On input $w\in \{0,1\}^n$, run $T(\epsilon, c, s, q)$ on $2^k w$ and accept if and only if $T$ accepts. The query complexity of $T_k$ is the same as the query complexity of $T$. If the input $w$ is a codeword in $C_{k}$, then by the definition of the lattice, $2^k w$ is a lattice vector, and $T_k$ will accept $w$ with probability at least $1-c$. If $d_1(w,C_k)\ge \eps n$, then by Lemma \ref{lemma:distance-to-lattice}, we have that $d_1(2^k w,L( \langle C_{i} \rangle)_{i=0}^{m-1})\ge \eps n$. Therefore, $T_k$ will reject $w$ with probability at least $1-s$.  Finally,  since $L\subseteq \Z^n$ we have that $d(w, \Z^n)\leq d(w, L)$ and so we could use $T$ to test membership in $\Z^n$. By Claim \ref{claim:lb-for-integer-lattice} 
testing $\Z^n$ requires $q=\Omega(\frac{1}{\eps}\log(1/s))$ queries.
\end{proof}

\begin{claim}\label{claim:lb-for-integer-lattice}
 Any test $T_k(\eps,c,s,q)$ for $\Z^n$ requires $q=\Omega(\frac{1}{\eps}\log(1/s))$ queries. 
\end{claim}
 
\begin{proof}
First, we use Yao's duality theorem \cite{ref:Yao77} which is a standard tool in proving lower bounds
and assume that the testing algorithm is, without loss of generality, deterministic
(but possibly adaptive).
We exhibit two distributions on the inputs which the algorithm is expected to
distinguish but cannot do so without making sufficiently many queries.
The \emph{yes case} distribution is the deterministic distribution on all-zeros
input (which is a lattice point). Given an input from this distribution, the algorithm should accept.
The \emph{no case} distribution would consists of the all-zeros vector but with
a uniformly random set $S$ of $2\eps n$ coordinate positions changed from $0$
to $1/2$. Indeed, all vectors on the support of this distribution are $\eps$-far
from the lattice.
Assuming that $q \leq n/2$ (otherwise there is nothing to prove), each time the
algorithm queries a position that has not been queries before, there is
at most a $4\eps$ chance that it hits any position in $S$ (even conditioned on
the past query outcomes). 
Thus the probability that the algorithm ever succeeds in finding a
position in $S$ is at most $1-(1-4\eps)^q$, which, on the other hand
by the soundness condition, should be at least $1-s$. Therefore,
the soundness error is at least $s\geq (1-4\eps)^q$ or, in other words,
in order to achieve a given $s$ we must have $q = \Omega(\frac{1}{\eps} \log(\frac1s))$.
 \end{proof}

In the case of code formula lattices generated from Reed-Muller codes of order $r$, we note that $n=2^r$. We need the following known lower bound on the query complexity of testing Reed-Muller codes. For completeness we reproduce the exact statement that we need in this work and include a proof. 

\begin{theorem}[\cite{AKKLR05}] \label{thm:RM:lb}
Let $T(\eps, c, s, q)$ be a (possibly $2$-sided and adaptive)  tester for 
the code $RM(k, r)$ where $k \leq r/(2 \log r)$, $\eps < 1/2-\Omega(1)$, and $c+s < 1-\Omega(1)$
(where $\Omega(1)$ hides arbitrarily small positive absolute constants).
Then, $q \geq2^{k}$.
\end{theorem}

\begin{proof}
\newcommand{\cC}{\mathcal{C}}
Using the reduction from 2-sided, adaptive testers to non-adaptive, 1-sided tests for any linear code (applicable to RM codes) of \cite{BHR05}, it is sufficient to focus on the latter tests.
Let $\cC$ be the code $RM(k, r)$. First we note that the length of the code is $R := 2^r$ and its dimension is 
\[
\log |\cC| = \sum_{i=0}^k \binom{r}{i} \leq 1+k r^k.
\]
Therefore, noting that $k \leq r/(2\log r)$,
\[
|\cC| = O(2^{k r^k}) = O(2^{r^{k+1}}) = O(2^{(\log R)\sqrt{R}})=2^{o(R)}.
\]
Let $V$ be the number of points in a Hamming ball of radius $\eps R$ in $\{0,1\}^R$. Thus we have
$V \leq 2^{h(\eps) R}$, where $h(\cdot)$ denotes the binary entropy function.
Let $S$ be the set of points in $\{0,1\}^R$ that have Hamming distance at most $\eps r$
with the code. Of course we have \[|S| \leq V |\cC| = O(2^{(h(\eps)+o(1)) R}).\]
Since $\eps < 1$ is a fixed constant, this implies $|S|/2^R = o(1)$.
Now we run the tester with the following two input distributions:
\begin{description}
\item[Case 1. ] The tester is given a uniformly random string in $\{0,1\}^R$ as the input.

\item[Case 2. ] The tester is given a uniformly random codeword of $\cC$ as the input.
\end{description}
Since the dual distance of $\cC$ is $2^{k+1}$, a standard coding theoretic fact implies
that a uniformly random codeword of $\cC$ is $t$-wise independent for $t=2^{k+1}$;
i.e., any local view of up to $t$ coordinates of the random codeword is exactly the uniform
distribution. Therefore, since the tester makes no more than $t$ queries, its output distribution 
is exactly the same in the above two cases. Let $p$ be the acceptance probability of the
tester with respect to the (common) output distribution. 
In order to satisfy completeness, the tester should accept with probability at least
$1-c$ in the second case. Therefore, we must have $p \geq 1-c$.

On the other hand, a uniform random string in $\{0,1\}^R$ is $\eps$-far from the code
with probability $1-o(1)$ according to the above bound on $|S|$. In the conditional world
where this string actually becomes $\eps$-far from the code, the acceptance probability
of the code would thus remain within $p(1\pm o(1))$. However, in this case the soundness
implies that the tester should accept with probability at most $s$, and thus,
$p \leq s(1+o(1))$. Thus the two distributions provided by the above two cases
would violate requirements of the local tester assuming that $c+s \leq 1-\Omega(1)$.
\end{proof}

\CoroConstrDLB*
\begin{proof}
Suppose we have a tester $T(\eps,c,s,q)$ for $L(\langle RM({k_i, r}) \rangle_{i=0}^{m-1})$. By Theorem \ref{theorem:lower-bound-tester-for-code}, we have
\[
q\ge \max_{i=0,1,\ldots,m-1} q_i(\eps,c,s).
\]
By Theorem \ref{thm:RM:lb}, it follows that $q=\Omega(2^{k_{m-1}})$.
\end{proof}
\section{Tolerant Testing Code Formula Lattices }
\label{sec:toltestingconstrD}
We first give a tolerant tester for testing membership in $\Z^n$. We will use this tester in the design of a tolerant tester for testing membership in lattices obtained from the code formula construction. 
\begin{lemma}
\label{lemma:tolerant-tester-for-Z}
Let $\eps_1, \eps_2 , c, s > 0$ such that $\eps_2 > \eps_1$ and $\gamma = \min \{c, s\}$
. There is a tolerant tester $T_Z(\eps_1, \eps_2, c, s, q_Z)$ for $\Z^n$ which uses $q_Z = O(1/(\eps_2-\eps_1)^2 \cdot \log(\frac1\gamma))$  queries.
\end{lemma}

\begin{proof}
The tester estimates the distance of the input from $\Z^n$ by querying $O(1/(\eps_2-\eps_1)^2 \log(\frac1\gamma))$ coordinates uniformly at random. If the estimated distance is at least $\frac{(\eps_1+\eps_2)}{2} n$, then it rejects, otherwise it accepts. The correctness and soundness follow from Chernoff bounds. 
We describe the test formally as follows:
\begin{enumerate}
\item Query $q := C/(\eps_2-\eps_1)^2 \cdot \log(\frac1\gamma)$ coordinates of the input $t$ uniformly at random, for some constant $C$ to be determined later. Let $I \subseteq [n]$ be the indices of the queried coordinates.
\item Let $\delta := \frac{ \sum_{i \in I} \lvert t_i - \round{t_i} \rvert }{q}$.
\item If $\delta \leq \frac{\eps_1+\eps_2}{2}$ then Accept.
\item Else Reject.
\end{enumerate}

Suppose $d(t,\Z^n)\le \eps_1 n$, then  $d(t, \Z^n) / n = \sum_{i=1}^n \lvert t_i - \round{t_i} \rvert / n \leq \eps_1$. Therefore, $\mathbb{E}[\delta]\le \eps_1$. By a Chernoff bound, it follows that 
$\Pr[ \delta -\eps_1 > \frac{\eps_2- \eps_1}{2} ] \leq e^{- q (\eps_2-\eps_1)^2/2}\le c$ for $q \geq C/(\eps_2-\eps_1)^2 \cdot \log(\frac1\gamma)$ and a constant $C > 0$.

Now suppose $d(t,\Z^n)>\eps_2 n$. Then, $d(t, \Z^n) / n = \sum_{i=1}^n | t_i - \round{t_i} | / n \geq \eps_2$. Again, by a Chernoff bound, and suitable choice of the constant $C$, it follows that $\Pr[ \eps_2-\delta \ge  \frac{\eps_2-\eps_1}{2} ]\leq e^{- q (\eps_2-\eps_1)^2/4}\le s$ for $q$ chosen as above.  
\end{proof}

We now describe a tolerant tester for code formula lattices. 

\thmTolerantTester*
\begin{proof}
We use the tolerant testers $T_i$ for the codes $C_i$ and the tolerant tester $T_Z$ for $\mathbb{Z}^n$ to construct a tolerant tester for $L$. 

Let $\round{t}$ denote the vector obtained by rounding each coordinate of $t$ to its nearest integer and for any vector $x$, let $x(j)$ denote the $j^{th}$ coordinate of $x$. Let $t_0, . . . , t_{m-1} \in  \{0, 1\}^n$ where $t_i(j)$ is the $(i + 1)^{th}$ least significant bit in the binary decomposition of the $j^{th}$ coordinate of $\round{t}$. Define $t_m = \frac{1}{2^m} (t - \sum_{i=0}^{m-1}2^i t_i) \in \mathbb{R}^n$. 
Therefore, $t$ can be written as $t = t_0 + 2t_1 + \cdots + 2^{m-1} t_{m-1} + 2^m t_m$, where $t_i \in \{0, 1\}^n$ for all $i \in [m-1]$ and $t_m \in \mathbb{R}^n$. Moreover, $t_m \in \mathbb{Z}^n$ if and only if $t \in \mathbb{Z}^n$. 

The tolerant tester $T(\eps_1, \eps_2, c, s, q)$ on input $t \in \mathbb{R}^n$  now proceeds as follows: Run $T_Z(\eps_1,\frac{ \eps_2}{ 2}, \frac{c}{m+1}, s, q_Z)$ on $t$ and $T_i(2 \eps_1,\frac{ \eps_2 }{ m2^{i+1}}, \frac{c}{m+1}, s, q_i)$ on $t_i$ for all $i\in \{0,1,\ldots,m-1\}$. Accept if and only if all tests accept. 
The query complexity of $T(\eps_1, \eps_2, c, s, q) $ is therefore:
\[ 
q(\eps_1, \eps_2, c, s) = \sum_{i=0}^{m-1} q_i + q_Z,
\]
where we recall that 
$q_Z$ is the query complexity of $T_Z(\eps_1, \frac{\eps_2}{2}, \frac{c}{m+1}, s, q_Z)$. From Lemma~\ref{lemma:tolerant-tester-for-Z}, we know that 
$q_Z = O(\frac1{(\eps_2- 2\eps_1)^2 }\log(\frac{1}{\gamma}))$, 
where $\gamma = \min\{ \frac{c}{m+1}, s\}$.
We now analyze the soundness and completeness of this test.

\subsubsection*{Soundness}
Suppose $d(t,L)\geq \eps_2 n$. We first show that either $t$ is far from $\Z^n$ or the closest integer vector to $t$ is far from the lattice.

\begin{claim} \label{claim:tineq1}
$d(t, L) \leq d(\round{t}, L) + d(t, \mathbb{Z}^n)$
\end{claim}

\begin{proof}
Let $u$ be the closest lattice vector to $\round{t}$. Then 
$$d(t, L ) \leq \| t - u \|_1 = \| ( t-\round{t} ) + (\round{t} - u) \|_1\leq  \| t-\round{t} \|_1+\|\round{t} - u \|_1$$
Since $\| t - \round{t} \|_1 = d(t, \mathbb{Z}^n)$, it follows that 
$ d(t, L) \leq d(\round{t}, L) + d(t, \mathbb{Z}^n)$. 
\end{proof}

Therefore, if $d(t, L) \geq \eps_2 n$, then from Claim~\ref{claim:tineq1}, either $d(\round{t}, L) \geq \eps_2 n/ 2 $ or $d(t, \mathbb{Z}^n) \geq \eps_2 n / 2$. If $d(t, \mathbb{Z}^n) \geq \eps_2n / 2$, then $T_Z$ rejects $t$ with probability at least $1 - s$. If $d(\round{t}, L) \geq \eps_2 n/ 2$, then from Claim~\ref{claim:triangle-inequality-decomposition} proved in Section~\ref{sec:testingconstrD} we can conclude that there exists some $i \in \{0,1,\ldots,m-1\}$ such that $2^i d(t_i, C_i) \geq \eps_2 n/ 2m$, and $T_i(t_i, 2\eps_1, \eps_2/m2^{i+1})$ will reject $t_i$ with probability at least $1- s$. Thus, if $d(t, L) \geq \eps_2n$, then $T$ rejects $t$ with probability with at least $1 - s$. 
  

\subsubsection*{Completeness}
Suppose $d(t,L)\le \eps_1 n$. Then $d(t,\Z^n)\le \eps_1 n$, since $L \subseteq \mathbb{Z}^n$.
So, $T_Z(\eps_1, \frac{\eps_2}{2}, \frac{c}{m+1}, s, q_Z)$ will accept $t$ with probability at least $1- \frac{c}{m+1}$. We now show that each $t_i$ is also close to the corresponding linear code $C_i$.  

For the sake of contradiction, suppose $d(t_i,C_i) > 2\eps_1 n$ for some $i\in [m-1]$.  
We will show that 
$d(t, L) > \eps_1n$. We do this in two steps.  First we show in Lemma~\ref{claim:far-from-code} that $d(\round{t}, L) >  2\eps_1n$. Then by Claim~\ref{claim:tineq2} and the fact that $d(t,\Z^n) \le \eps_1 n$, we have that 
$d(t, L) > \eps_1n$, a contradiction.

\begin{lemma} \label{claim:far-from-code} 
If $d(t_i, C_i) > 2\eps_1n$ for some $i \in [m-1]$, then $d(\round{t}, L) > 2\eps_1n$. 
\end{lemma}

\begin{proof}
Let $v = \sum_{i=0}^{m-1} 2^i v_i + 2^m v_m \in L$ be the closest lattice vector to $\round{t}$. By definition of the lattice, each $v_i \in C_i$ for $i \in [m-1]$ and $v_m \in \mathbb{Z}^n$. Consider the vectors $t_0,t_1,\ldots, t_m$ as defined above (for which $\round{t} = \sum_{i=0}^{m-1} 2^i t_i + 2^m t_m$). So, each $t_i \in \{0, 1\}^n$ and $t_m \in \mathbb{Z}^n$. The following property of vectors with bounded entries will be used to prove the claim. 
\begin{claim}\label{claim:property-of-binary-vectors}
Let $a_0, a_1, \ldots, a_{m-1} \in \{-1, 0, +1\}^n$ and $a_m \in \mathbb{Z}$. Define $u = a_0 + 2a_1 + \cdots + 2^m a_m$.  If there exists some $k \in [m-1]$ such that $\| a_k \|_1 >  s$, then $\|u\|_1 > s$.  
\end{claim}
\begin{proof}
Since $\|a_k\|_1 > s$, and $a_k \in \{-1, 0, +1\}^n$, there exist at least $s$ coordinates such that $ \lvert a_k(i) \rvert = 1$. Let $S$ be the set of those indices, $S = \{ i \in [n]\colon \lvert a_k(i) \rvert = 1\}$. Since $\|a_k\|_1 > s$, we know that $|S| > s$. We now show that for all $i \in S, \lvert u(i) \rvert \geq 1$. Therefore, $\| u \|_1 > s$.  

Let $i \in S$. For each such coordinate, we can express $u(i) = \sum_{j = 0}^m 2^j a_j(i)$. Let $h \in [k]$ be the smallest integer such that $a_h(i) \neq 0$. We know that such $h$ exists since $a_k(i) \neq 0$. Therefore, we know that $u(i) \mmod 2^{h+1} (= a_h)$, is non-zero. Therefore, $u(i)$ is also non-zero. Since $u(i) \in \mathbb{Z}$, we have that $| u(i) | \geq 1$. 

Therefore, $| u(i) | \geq 1$ for all $i \in S$. and $\|u\|_1 \geq |S| >  s$.
\end{proof} 

Define $a_i = (t_i - v_i)$ for all $i \in [m]$. We note that each $a_i \in \{ -1, 0, +1\}^n$ for $i \in [m-1]$ and that $a_m \in \mathbb{Z}^n$. The proof now follows from Claim~\ref{claim:property-of-binary-vectors} for $s = 2\eps_1 n$. 
\end{proof}

The next claim is a straightforward application of the triangle inequality. 
\begin{claim}\label{claim:tineq2}
$d(t, L) \geq d(\round{t}, L) - d(t, \mathbb{Z}^n)$
\end{claim} 

\begin{proof}
Let $u$ be the closest lattice vector to $t$. 
$$d(\round{t}, L) \leq \| \round{t} - u \|_1 = \| \round{t} - t + t - u \|_1.$$
By the triangle inequality, we have $\| \round{t} - t + t - u \|_1 \leq \| \round{t} - t \|_1 + \| t - u\|_1$.  Since $u$ is  the closest lattice vector to $t$, $d(t, L) = \| t - u \|_1$. Also,  $\| \round{t} - t \| = d(t, \mathbb{Z}^n)$, Therefore, $d(\round{t}, L) \leq d(t, L) + d(t, \mathbb{Z}^n)$.
\end{proof}

Thus, if $d(t,L)\le \eps_1 n$, then 
$T_Z$ accepts $t$ with probability at least $1-\frac{c}{m+1}$ and each code tester $T_i$ accepts $t_i$  with probability at least $1-\frac{c}{m+1}$.
Therefore, from the union bound, $T$ accepts $t$ with probability at least $1- \sum_{i=0}^m \frac{c}{m+1} = 1-c$. 
\end{proof}

We next instantiate Theorem \ref{theorem:tolerant-tester-code-formula} for code-formula lattices obtained from Reed-Muller codes. We first recall a simple observation made in \cite{PRR06} that any local test with individual queries uniformly distributed is also a tolerant test.

\begin{claim}[\cite{PRR06}]\label{claim:prr}
If a code $C\subseteq \{0,1\}^n$ has a one-sided local test $T(\eps, 0, 1/3, q)$ whose queries are each uniformly distributed, then $C$ has a tolerant test $T(\eps_1, \eps_2, 1/3, 1/3, q)$, with $\eps_1\leq \frac{1}{3q}$ and $\eps_2\geq \eps$.
\end{claim}

Using Claim \ref{claim:prr} and Theorem \ref{theorem:RM-codes-upper-bound}, and by appropriately amplifying the success probability, we get a tolerant test for Reed-Muller codes. 
\begin{corollary}\label{coro:RMtolerant}
For any $k,r, c, s>0$ and $\gamma = \min\{c, s\}$,  there exists a tolerant test $T(\eps_1, \eps_2, c, s, q)$ for  $RM(k,r)$ such that $\eps_1\leq c_1 \frac{1}{2^k}$, $\eps_2\geq c_2 \frac{1}{2^k}$ and $q=O(2^k \log(\frac1\gamma))$, for some $c_1, c_2>0$.
\end{corollary}

\begin{proof}
By Theorem~\ref{theorem:RM-codes-upper-bound} we know that there is a 1-sided tester $T(\eps, 0, 1/3, q)$ for $RM(k,r)$ and $\eps = O(1/2^k)$ with query complexity $q = O(2^k)$ . From Claim~\ref{claim:prr}, we know that we can obtain a tolerant tester $T(\eps_1, \eps_2, 1/3, 1/3, q)$ with $O(2^k)$ queries for any $\eps_1 \leq c_1/2^k$ and $\eps_2 \geq c_2/2^k$. 
By independently repeating the tester multiple times and taking majority vote to amplify the success probability, for any $0<c,s\leq1$ and $\gamma = \min\{c, s\}$ we get a tolerant tester $T(\eps_1, \eps_2, c, s, q)$ for $RM(n,k)$ with $q=O(2^k \log(\frac1\gamma))$ queries. 
\end{proof}

Using Corollary \ref{coro:RMtolerant} and Theorem \ref{theorem:tolerant-tester-code-formula} we obtain the following immediate corollary. 

\CoroTolerantRM*
\begin{proof}
From Corollary \ref{coro:RMtolerant}, every $RM_{k_i}$ has a tolerant tester $T_i(2\eps_1, \frac{\eps_2}{m2^{i+1}}, \frac{1}{3(m+1)},  \frac{1}{3}, q_i)$ with query complexity $q_i = O(2^{k_i}  \log(m+1))$ for $2\eps_1 \leq \frac{c_1}{2^{k_i}}$ and $\frac{\eps_2}{m2^{i+1}} \geq \frac{c_2}{2^{k_i}}$ for some constants $c_1, c_2 > 0$. 

Using Theorem  \ref{theorem:tolerant-tester-code-formula}, we therefore conclude that $L$ has a tolerant tester $T(\eps_1, \eps_2, \frac13, \frac13, q)$ with query complexity 
$q = O( \frac{1}{(\eps_2 - 2\eps_1)^2}  \log(m+1) )  + \sum_{i=0}^{m-1}O( 2^{k_i}  \log(m+1) ) = O(2^{k_{m-1}} \log m )$ 
for $2\eps_1 \leq \min_i \{ \frac{c_1}{2^{k_i}} \}$ and $\frac{\eps_2}{m2^{i+1}} \geq \max_i \{ \frac{c_2}{2^{k_i}} \}$.
\end{proof}
\section{Reducing an arbitrary test to a non-adaptive linear test}
\label{sec:reduction}

In this section we sketch the proof of Theorem~\ref{theorem:2-sided-to-1-sided-non-adaptive-linear}. Throughout this section, we focus on full-rank integral lattices. 
Given a $2$-sided adaptive $\ell_p$-tester $T(\eps,c,s,q)$, with $q=q_T(\eps,c,s)$ for an integral lattice $L$, we construct a non-adaptive linear $\ell_p$-tester $T'(\eps,0,c+s,q)$ with query complexity $q'=q_T(\eps/2,c,s)+O((1/\eps^p)\log{(1/s)})$. 
We reduce the inputs to a bounded set using the following property of integral lattices.

 \begin{fact} \cite{MiccLN12}
Given any full rank integral lattice $L$, there exists $d\in \Z$  such that $ d\cdot \mathbb{Z}^n\subseteq L$. In particular $ \lvert \det(L) \rvert \cdot \mathbb{Z}^n \subseteq L$ for any lattice (where $\det(L)$ denotes the determinant of a lattice,  a parameter that can be computed given a basis of the lattice).   
For instance, we can take $d = 2^m$ for the lattices of height $m$ obtained using the code formula construction.
\end{fact}

Let $V=L \mmod d$ embedded in $\Z^n$ (i.e., we treat $V$ as a set of vectors in $\Z^n$ each of which is obtained by taking coordinate-wise modulo $d$ of some lattice vector). Thus, $V\subseteq \calz_d^n$. We will need the following properties of $V$, which we prove, for the sake of completeness. 

\begin{restatable}{prop}{propModPreservesDistance} \label{prop:mod-preserves-distance}
Let $L\subseteq \Z^n$ be a full-rank lattice, $d\in \Z_+$ such that $d \Z^n\subseteq L$, and let $V=L \mmod d \subseteq \Z^n$. 
Then $V$ satisfies the following properties:
\begin{enumerate}
\item $v\in L$ if and only if $v \mmod d \in V$.
\item  $V = L \cap {\cal Z}_d^n.$ 
\item $(v+V) \mmod d\subseteq V$ if and only if $v\in L$.
\item For any $v \in \Z^n$, $d_p(v,L)=d_p(v \mmod d,L)$.
\end{enumerate}
\end{restatable}

\begin{proof}
\begin{enumerate}
\item If $v\in L$, then $v \mmod d\in V$ by definition. For the opposite direction, let $v\in \Z^n$ be such that $u=v \mmod d\in V$. 
Then by the definition of $V$ there exists $v'\in L$ such that $v'=u=v \pmod d$. Then $v-v'\in d\Z^n\subseteq L$, and so $v\in L$.

\item By definition $L\cap \calz_d^n\subseteq V$. To show that $V\subseteq L$ note that by 1), if $v\in V$ there exists $v'\in L$ such that $v'=v \mmod d$. As before, this implies that $v\in L$.

\item This statement follows by the fact that $V\subseteq L$ and from the fact that lattices are closed under addition.

\item Note that $d_p(v, L)=\min_{u\in L} d_p(u,v)=\min_{u\in L} \|v-u\|_{p}$. If $v=d v_1+v_2$, since $d v_1\in d\Z^n\subseteq L$, it follows that  $\min_{u\in L} \|v-u\|_{p}=\min_{u\in L} \|v_2-u\|_p$, since a lattice is closed under addition.

\end{enumerate}

\end{proof}
Theorem \ref{theorem:2-sided-to-1-sided-non-adaptive-linear} will immediately follow by combining Lemmas \ref{lemma:2-sided-to-1-sided-linear}, \ref{lemma:adaptive-to-non-adaptive}, \ref{lemma:bounded-targets}, and \ref{lemma:integer-targets} which will be proved in the subsequent sections. 


\begin{restatable}{lem}{TwoSidedToLinear}
\label{lemma:2-sided-to-1-sided-linear}
Suppose a full-rank lattice $L\subseteq \Z^n$ with $d\Z^n\subseteq L$ for $d\in \Z_+$ has an adaptive 2-sided $\ell_p$-tester $T(\epsilon, c, s, q)$ for inputs from the domain ${\cal Z}_d^n$. Then $L$ has an adaptive linear $\ell_p$-tester $T'(\epsilon, 0, c+s, q)$ for inputs from the domain ${\cal Z}_d^n$.
\end{restatable}


\begin{restatable}{lem}{AdaptiveToNonAdaptive}
\label{lemma:adaptive-to-non-adaptive}
Suppose a full-rank lattice $L\subseteq \Z^n$ with $d\Z^n\subseteq L$ for $d\in \Z_+$ has  an adaptive linear $\ell_p$-tester $T(\epsilon, 0, s, q)$  for inputs from the domain ${\cal Z}_d^n$. Then $L$ has a non-adaptive linear $\ell_p$-tester $T'(\epsilon, 0, s, q)$ for inputs from the domain ${\cal Z}_d^n$.
\end{restatable}


\begin{restatable}{lem}{BoundedTargets} \label{lemma:bounded-targets}
Let $L\subseteq \Z^n$ be a full-rank lattice with $d\Z^n\subseteq L$ for $d\in \Z_+$. Then, $L$ has a non-adaptive linear $\ell_p$-tester $T(\epsilon, 0, s, q)$  for inputs from the domain ${\calz}_d^n$ if and only if $L$ has a non-adaptive linear $\ell_p$-tester $T'(\epsilon,0, s, q)$ for inputs from the domain $\Z^n$.
\end{restatable}


\begin{restatable}{lem}{RealTargets}\label{lemma:integer-targets}
Suppose a full-rank lattice $L\subseteq \Z^n$ has a non-adaptive $\ell_p$-tester $T(\epsilon, c, s, q)$ for inputs from the domain ${\Z}^n$.  
Then there exists a non-adaptive $\ell_p$-tester $T'(\epsilon,c, s, q')$ for inputs in $\R^n$ with query complexity $q'=q(\eps/2,c,s) + O((1/\eps^p)\log{(1/s)})$. Moreover, if $T$ is a linear tester, then so is $T'$.
\end{restatable}

~

The proof of Lemma \ref{lemma:integer-targets} uses the following tester for integer lattices which is based on querying a random collection of coordinates and verifying whether all of them are integral. 

\begin{restatable}{lem}{IntegerLatticeTester}\label{lemma:integer-lattice-tester} 
For every $0 < \eps \leq 1$ and every $0 < s \leq 1$, there exists a non-adaptive linear $\ell_p$-tester $T_p(\epsilon,0,s, q_Z)$ for $\Z^n$ with query complexity  
\[q_Z=O\left(\frac{1}{\eps^p} \log{\frac1s}\right).\]
\end{restatable}


\subsection{2-sided to Linear Tester}
In this section, we prove Lemma~\ref{lemma:2-sided-to-1-sided-linear}.  
Given a 2-sided adaptive tester $T(\epsilon, c, s, q)$ for inputs $x$ from the domain ${\cal Z}_d^n$, we build an adaptive linear (thus one-sided) test $T'(\epsilon, 0, c+s,q)$ for inputs from the same domain $\calz_d^n$ with the same query complexity as that of $T$ in this section. 

For an index set $J\subseteq [n]$ and a vector $w\in {\cal Z}_d^n$, let $X(w,J):=\{x\in {\cal Z}_d^n \mid (\forall\ i\in J)\ x_i=w_i\ \}$. For a subset of coordinates $J\subseteq [n]$ and a vector $w\in \calz_d^n$, we say that \emph{there exists a dual witness for $X(w,J)$} if there exists $\alpha \in L^{\bot}_J$ such that $\inprod{\alpha}{w}\notin \Z$. That is, a dual witness $\alpha$ is a dual vector entirely supported on $J$ that proves none of the vectors in $X(w,J)$ (and thus $w$) can be in the lattice.
Recall that $V := L \mmod d$ where $d\Z^n \subseteq L$.

If the input vector $x$ is from the domain ${\cal Z}_d^n$, then each coordinate of the input has $d$ possible choices. Thus, any 2-sided adaptive tester $T(\epsilon, c, s,q)$ for inputs from the domain $\mathbb{Z}_d^n$, can be viewed as a distribution over deterministic decision trees with each leaf being labeled $1$ if accepting and $0$ if rejecting. Therefore we will express the tester as $T = (\Upsilon_T, D_T)$, where $\Upsilon_T$ is the set of all decision trees (with at most $q$ queries) 
and $D_T$ is a distribution over $\Upsilon_T$. 

Let $l$ be a leaf of a decision tree. We denote the coordinates queried along the path to $l$ by $\var(l)$. We denote the vector that is consistent with the queried coordinates along the path to $l$ and has zeros in the non-queried coordinates by $s_l$. 
Let us define $V_l$ to be the set of lattice vectors $u$ which are consistent with the queries along the path to $l$. Similarly, let $V_l^x$ be the set of vectors in $(x+ V) \mmod d$ which are consistent with the queries along the path to $l$, i.e., $V_l  = X(s_l,\var(l)) \cap V$ and $V_l^x = X(s_l,\var(l)) \cap ((x + V) \mmod d)$. We need the following claim about the sizes of $V_l$ and $V_l^x$.

\begin{claim}
\label{equal size}
For every leaf $l$ in the decision tree $\Gamma$, if both $V_l$ and $V_l^x$ are non-empty, then $\lvert V_l \rvert = \lvert V_l^x \rvert $.
\end{claim}
\begin{proof}
Let $U$ denote the set of all the lattice vectors in ${\cal Z}_d^n$ which have all $0$'s in the positions queried along the path to $l$. We know that $U$ is non-empty because the all zeros vector is in $U$. 

For every $v \in V_l$ and $u\in U$, we have that $(v + u) \mmod d$ is also in $V_l$ since we are only adding $0$'s at the queried coordinates. Similarly, for every vector $v' \in V_l^x$ and $u\in U$, we have that $(v' + u) \mmod d$ is also in $V_l^x$. Therefore, we know that $(U + v) \mmod d \subseteq V_l$ for every $v \in V_l$ and similarly, $(U + v') \mmod d \subseteq V_l^x$ for every vector $v' \in V_l^x$. 

Further, for every two vectors $u, v \in V_l$, we have that $(u-v) \mmod d$ is in $U$ and since $u$ and $v$ are both consistent along the path to $l$, the vector $u-v$ has all zeros at the queried coordinates. So, $(u-v) \mmod d \in U$. Therefore, $(V_l - v) \mmod d \subseteq U$ for every $v\in V_l$ and hence $V_l \subseteq (U + v) \mmod d$ for every $v\in V_l$. 
Similarly, for every vector $v' \in V_l^x$, we have that $(V_l^x - v') \mmod d \subseteq U$ and hence $V_l^x \subseteq (U + v') \mmod d$.  

Therefore, if $V_l$ and $V_l^x$ are non-empty, then $(U + v) \mmod d = V_l $ for every vector $v \in V_l$ and $(U + v') \mmod d = V_l^x$ for every vector $v' \in V_l^x$. Hence, $ \lvert V_l \rvert = \lvert U \rvert = \lvert V_l^x \rvert$. 
\end{proof}

We now show that if a linear test accepts, then there exists a lattice vector that is consistent with the queried coordinates. In other words, if there is no dual witness then there is a lattice vector that is accepted by the test. 

In the following, let $\proj_J(u)\in \R^{|J|}$ denote the projection of vector $u$ to the coordinates in $J$ and $\proj_J(S)$ denote the set of vectors obtained by projecting the vectors in $S$ to the coordinates in $J$.
We note that the projection of a rational lattice to a set of coordinates gives a lattice again.

\begin{prop}\label{prop:no-dual-witness-implies-lattice-point-exists}
Let $J\subseteq [n]$, $w\in \calz_d^n$. If $\inprod{\alpha}{\proj_J(w)}\in \Z$ for every $\alpha \in \proj_J(L^{\bot}_J)$, then $V\cap X(w,J)\neq \emptyset$.
\end{prop}
\begin{proof}

 We recall that the dual of a projection of a lattice is the set of vectors in the projected space which have integral dot products with all points in the projected lattice. The following proposition shows that the dual of a projected lattice is the projection of the set of vectors in the dual lattice whose support is contained in the projection. 
 
\begin{prop}\label{prop:projection-of-dual-is-dual-of-projection}
Let $J\subseteq [n]$. Then 
\[
\left(\proj_J(L)\right)^{\bot}=\proj_J\left(L^{\bot}_J\right).
\]
\end{prop}
\begin{proof}
Let $\alpha_J\in \proj_J(L^{\bot}_J)$. Let us extend the vector $\alpha_J$ to $\alpha\in \R^n$ by setting the coordinates that are not in $J$ to zero. We note that $\alpha\in L^{\bot}_J$. Hence $\inprod{\alpha}{x}\in \Z$ for every $x\in L$. Therefore $\inprod{\alpha_J}{x_J}\in \Z$ for every $x_J\in \proj_J(L)$. Thus, $\alpha_J\in (\proj_J(L))^{\bot}$.

Let $\alpha_J \in (\proj_J(L))^{\bot}$. Then for every $v_J\in \proj_J(L)$, we have $\inprod{\alpha_J}{v_J} \in \Z$. Consequently for every $v\in L$, we have $\inprod{\alpha_J}{\proj_J(v)}\in \Z$. Let us extend the vector $\alpha_J$ to $\alpha\in \R^n$ by setting the coordinates that are not in $J$ to zero. Then $\inprod{\alpha}{v}\in \Z$ for every $v\in L$. Therefore $\alpha\in L^{\bot}$ and hence $\alpha_J\in \proj_J(L^{\bot}_J)$. 
\end{proof}

We have that $\inprod{\alpha}{\proj_J(w)}\in \Z$ for every $\alpha\in \proj_J(L^{\bot}_J)$. Therefore $\proj_J(w)\in (\proj_J(L^{\bot}_J))^{\bot}$. By Proposition \ref{prop:projection-of-dual-is-dual-of-projection}, we have that $\proj_J(w)\in \proj_J(L)$. Hence, there exists $x\in X(w,J)\cap L =X(w,J)\cap V$.
\end{proof}

 Note that it is possible to determine if there exists a dual witness for $X(w,J)$ and if so, find one efficiently as shown in Proposition \ref{prop:dual-witness-can-be-found-efficiently}.
\begin{prop}\label{prop:dual-witness-can-be-found-efficiently}
Given $w\in {\cal Z}_d^n$ and $J\subseteq [n]$, we can find a dual witness for $X(w,J)$ if one exists or confirm that no dual witness for $X(w,J)$ exists in time $O(|J|^{\omega})$, where $O(m^{\omega})$ is the time to compute the inverse of a $m\times m$ real matrix. 
\end{prop}
\begin{proof}
A basis for $\proj_J(L)$ can be obtained by projecting the basis for $L$. Now a basis for the dual of the projected lattice, namely $\proj_J(L)^{\bot}=\proj_J(L^{\bot}_J)$, can be computed in time $O(|J|^{\omega})$. We observe that for every $\alpha\in L_J^{\bot}$, we have $\inprod{\alpha}{w}\in \Z$ if and only if for every basis vector $b$ of $\proj_J(L^{\bot}_J)$, we have $\inprod{b}{\proj_J(w)}\in \Z$. Hence it is sufficient to only verify the inner product of $\proj_J(w)$ with the basis vectors of $\proj_J(L^{\bot}_J)$.
\end{proof}

We now have the ingredients needed to prove Lemma \ref{lemma:2-sided-to-1-sided-linear}.

\TwoSidedToLinear*
\begin{proof}
We first relabel the decision tree according to the rule required for a linear test: Given a decision tree $\Gamma$ for the tester $T$, we say that it is \emph{optimally labeled} if the label of any leaf $l$ is $0$ 
whenever there exists a dual witness  
for $X(s_l,\var(l))$ 
and $1$ otherwise. 
We denote the tree obtained from $\Gamma$ by optimally relabeling to be $\Gamma_{OPT}$ (the relabeling for a given leaf of a tree $\Gamma$ can be done efficiently by Proposition \ref{prop:dual-witness-can-be-found-efficiently}). 
We build a tester $T'$ as follows:
\begin{enumerate}
\item On input $x \in {\cal Z}_d^n$, choose a tree $\Gamma$ according to $D_{T}$.
\item Choose a uniformly random vector $v$ in $V$ (recall that $V := L \mmod d$).
\item Answer according to the relabeled decision tree $\Gamma_{OPT}$ on input $(x + v)\mmod d$.
\end{enumerate}

It is clear that $T'$ is a linear test and has the same query complexity as that of $T$. 
We now show that the probability of acceptance by $T'$ of any vector $w$ which is $\epsilon$-far from $L$, does not exceed $c+s$. Let us define the following for a tester $\bar{T}$:
\begin{align*}
\rho^{\bar{T}} &:= \underset{y \in V} {\operatorname{avg} } ~Pr[ \bar{T}(y) = 1], \\
\rho^{\bar{T}}_x &:= \underset{y \in (x+V) \mmod d }{\operatorname{avg} } Pr[ \bar{T}(y) = 1].
\end{align*}
Due to the randomness in the choice of the tester $T'$, we have  
\begin{align*}
\rho^{T'} &= Pr[ T'(x) = 1 \mid x \in V],   \\
\rho^{T'}_x &= Pr[ T'(x) = 1].
\end{align*}

Since $T'$ is a $1$-sided tester, we have that $\rho^{T'} = 1$. Since $T$ accepts lattice vectors with probability at least $1-c$, we have $\rho^{T} \geq 1  - c$. Let $x \in {\cal Z}_d^n$ be $\epsilon$-far from $L$. For every $v\in V$, we have that $(x+v) \mmod d$ is also $\epsilon$-far from $L$ by Proposition \ref{prop:mod-preserves-distance}. Therefore, $\rho^{T}_x \leq s$. Using Claim~\ref{prob ineq}, we have 
\[ \rho^{T'}_x  \leq \rho^{T'} - \rho^{T} + \rho^{T}_x  \leq 1 - (1 - c) + s = c+s. \]
\end{proof}

\begin{claim}
\label{prob ineq}
For every $x\in {\cal Z}_d^n$,
\[ \rho^{T'}_x  \leq \rho^{T'} - \rho^{T} + \rho^{T}_x. \]
\end{claim}
\begin{proof}
Let $x$ be a vector in ${\cal Z}_d^n$. We analyze the effect of relabeling a single leaf $l$ of the decision tree $\Gamma$. We show that relabeling $l$ optimally preserves the claim and hence by repeated relabeling, we can deduce the claim.

Case (i). There exists a dual witness for $X(s_l,\var(l))$. Then the leaf $l$ is relabeled from $1$ to $0$. 
If input $y\in X(s_l,\var(l))$, then $y$ cannot be a lattice vector (if $y$ is a lattice vector, then there cannot exist a dual witness for $X(s_l,\var(l))$).  
Therefore, the probability of acceptance of lattice vectors is not changed due to relabeling, i.e., 
$\rho^{T'} = \rho^{T}$. If the leaf $l$ is reached for input $y\in {\cal Z}_d^n\setminus V$, then $T'$ rejects. Thus, relabeling does not increase the probability of acceptance of non-lattice vectors, i.e., 
$\rho^{T'}_y \leq  \rho^{T}_y $. Therefore, $ \rho^{T'}_x  \leq \rho^{T'} - \rho^{T} + \rho^{T}_x $ holds for this case. 

Case (ii). There does not exist a dual witness for $X(s_l,\var(l))$. Then the leaf $l$ is relabeled from $0$ to $1$.

The set of vectors in $V_l \cup V_l^x$ were rejected by $T$ and, after optimal relabeling of the leaf $l$, are now accepted by $T'$. The rest of the vectors in $V$ and $(x+V) \mmod d$ are rejected/accepted equally by both $T$ and $T'$.

Now, if $y$ was a lattice vector, then the probability of accepting a lattice vector increases because of the relabeling of $l$. Among the vectors in $V$, the vectors in $V_l$ are precisely the ones which were rejected before relabeling and are now accepted after relabeling. Since we average over all possible vectors $y\in V$ in the definition of $\rho^{T}$, the fractional change in the acceptance probability given that $T'$ and $T$ chose the decision tree $\Gamma$ is exactly $|V_l|/|V|$. Therefore, 
\[ \rho^{T'} = \rho^{T} + D_{T}(\Gamma) \frac{ \lvert V_l \rvert }{ \lvert V \rvert}. \]

Among the vectors in $(x+V) \mmod d$, the vectors in $V_l^x$ are the only vectors which were rejected before relabeling and are now accepted after relabeling. Thus, the fractional change in the acceptance probability of $(x+v)\mmod d$ given that $T'$ and $T$ chose the decision tree $\Gamma$ is exactly $ \lvert V_l^x \rvert/ \lvert V \rvert$. Therefore, 
\[ \rho^{T'}_x = \rho^{T}_x + D_{T}(\Gamma) \frac{  \lvert V_l^x \rvert }{\lvert V \rvert }. \]

Combining the two equations, we get
 \[ \rho^{T'}_x  =\rho^{T'}  - \rho^{T} +\rho^{T}_x  +  \frac{D_{T}(\Gamma)}{ \lvert V \rvert} ( \lvert V_l^x \rvert - \lvert V_l \rvert ). \]
Using Claim~\ref{equal size}, we know that $\lvert V_l^x \rvert \leq \lvert V_l \rvert$ if $V_l$ is non empty. Since there does not exist a dual witness for $l$, by Proposition \ref{prop:no-dual-witness-implies-lattice-point-exists}, we have that $V_l$ is non-empty. Hence the claim follows. 
\end{proof}


\subsection{Adaptive to Non-adaptive}
In this section we show that given an adaptive linear tester for a lattice, we can construct a non-adaptive linear tester from it without increasing the query complexity or the acceptance probability of non-lattice vectors. 

\AdaptiveToNonAdaptive*
\begin{proof}
Let $T(\epsilon, 0, s,q)$ be an adaptive linear tester for inputs from the domain ${\cal Z}_d^n$ with query complexity $q$. We construct a non-adaptive linear tester $T'(\epsilon, 0, s,q)$ for inputs from the domain ${\cal Z}_d^n$ as follows:
\begin{enumerate}
\item On input $x \in {\cal Z}_d^n$, choose a random vector $v \in V$. 
\item Run $T$ on input $v$. Let $J$ denote the set of coordinates that are queried.
\item Query $x$ on all the coordinates in $J$.
\item Reject if and only if there exists a dual witness for $X(x,J)$.
\end{enumerate}

We note that $T'$ is a linear test and the query complexity of $T'$ is the same as the query complexity of $T$. Since the queries depend only on a random $v\in V$ and not on the input $x$, the test $T'$ is non-adaptive. It remains to bound the acceptance probability of non-lattice vectors by $T'$. We will show that there is no dual witness for $X(x,J)$ if and only if there exists a vector $y\in (x+V)\mmod d$ that is consistent with the queried coordinates of $v$. 
As a consequence, we will show that the probability that $T'$ accepts $x$ is identical to the average acceptance probability of $x+v$ for random vectors $v\in V$ by $T$. Before analyzing the acceptance probability, we introduce a few notations and observations. 

For a decision tree $\Gamma\in \Upsilon_{T}$, we denote the set of leaves of $\Gamma$ which are labeled $1$ by $l_1(\Gamma)$. For a leaf $l$ of $\Gamma$ and a vector $x \in {\cal Z}_d^n$,  
let $I_l^x$ be a boolean (indicator) variable which takes a value of $1$ if and only if $\inprod{\alpha}{x}\in \Z$ for every $\alpha \in L^{\bot}_{var(l)}$. 

Let $\bar{\Gamma}$ be the decision tree chosen by the tester $T'$ on input $x$. The random vector $v\in V$ chosen by $T'$ corresponds to a leaf labeled $1$ in $\bar{\Gamma}$. This is because $T$ is a linear test and hence a lattice vector $v$ cannot have any dual witness. Therefore, $v\in V_{\bar{l}}$ for some $\bar{l}\in l_1(\bar{\Gamma})$. Since $T'$ is a linear test it is clear that $T'$ accepts $x$ if and only if $I_{\bar{l}}^x=1$.

\begin{claim}
\label{coset constraints}
Let $l$ be a leaf of a decision tree $\Gamma\in \Upsilon_{T}$, $x\in {\cal Z}_d^n$ and $y\in (x+V)\mmod d$.We have that $I_l^x=1$ if and only if $I_l^y=1$.
\end{claim}
\begin{proof}
If $y\in (x+V) \mmod d$, then $x-y\in L$. If $x, y\in \Z^n$ belong to the same coset of $L$, then for every $a\in L^{\bot}$, we have that $\inprod{x}{a}\in \Z$ if and only if $\inprod{y}{a}\in \Z$. Therefore, there exists $\alpha\in L^{\bot}_{var(l)}$ such that $\inprod{\alpha}{x}\notin \Z$ if and only if there exists $\alpha\in L^{\bot}_{var(l)}$ such that $\inprod{\alpha}{y} \notin \Z$. Hence $I_l^x=1$ if and only if $I_l^y=1$ for every $y\in (x+V) \mmod d$.
\end{proof}

\begin{claim}
\label{equal size 1}
Let $x\in {\cal Z}_d^n$ and $l$ be a leaf of a decision tree $\Gamma\in \Upsilon_{T}$ such that $l\in l_1(\Gamma)$. Then  
$\lvert V_l^x \rvert = I_l^x \lvert V_l \rvert$.
\end{claim}

\begin{proof}
We know that for every leaf $l$ which is labeled $1$, the set $V_l$ is non-empty since $T$ is a linear tester (using Proposition \ref{prop:no-dual-witness-implies-lattice-point-exists}). By Claim~\ref{equal size} we know that $| V_l^x | = |V_l |$ if $V_l^x$ is also non-empty.  Therefore it is sufficient to show that $I_l^x = 1$ if and only if $V_l^x$ is non-empty 

If $V_l^x$ is non-empty, then by definition, there is a vector $y \in (x+V) \mmod d$ which is consistent with all the queries along the path to $l$. Since $l$ is labeled $1$, we know that $T$ accepts $y$. Since $T$ is a linear tester, this implies that there does not exist an $\alpha\in L^{\bot}_{\var(l)}$ such that $\inprod{\alpha}{x}\notin \Z$. Hence $I_l^y=1$. By Claim~\ref{coset constraints}, we know that $I_l^x$ is also $1$.

If $I_l^x=1$, then for every $\alpha\in L^{\bot}_{\var(l)}$, we have $\inprod{\alpha}{x}\in \Z$. By Proposition \ref{prop:no-dual-witness-implies-lattice-point-exists}, there exists a vector $v\in V\cap X(x,\var(l))$. Hence, we have a vector $v\in V$ whose entries are identical to that of $x$ at the coordinates in $\var(l)$. We observe that the vector $(x-v)\mmod d$ has all $0$ entries at the coordinates in $\var(l)$. Further, $V_l$ is non-empty since $l$ is labeled $1$. Let $u\in V_l$. Then $((x-v)+u)\mmod d$ is consistent with all queries along the path to $l$, and is in $(x+V)\mmod d$. Therefore $V_l^x$ is non-empty.
\end{proof}

We now show that the acceptance probability of $T'$ is equal to the average acceptance probability of $T$.  Let 
\[ \rho_x := \underset{v \in V }{\operatorname{avg} } \Pr[ T((x + v)\mmod d) = 1]. \]
We note that this quantity is $1$ if $x \in V$ and is at most $s$ if $x$ is $\epsilon$-far from the lattice $L$.  The following claim shows that $T'$ accepts an input vector $x$ with probability $1$ if $x\in V$ and with probability at most $s$ if $x$ is $\epsilon$-far from the lattice $L$.
\end{proof}

\begin{claim}
\label{main}
Let $x \in {\cal Z}_d^n$. Then
$\Pr[T'(x) = 1] = \rho_x$.
\end{claim}
\begin{proof}
The average acceptance probability of $T$ can be viewed as follows: we pick a decision tree $\Gamma$ according to $D_T$. Then we pick a leaf $l$ labeled $1$ with probability proportional to the fraction of vectors in $x+V\mmod d$ that are consistent with the queries along the path to $l$. Therefore, 
\[ \rho_x = \sum_{\Gamma \in \Upsilon_{T}} D_{T}(\Gamma) \left ( \sum_{l \in l_1(\Gamma)} \frac{| V_l^x | }{ | V |} \right ).\] 

We have seen that $T'(x)  = 1$ if and only if for the random vector $v\in V$ chosen by $T'$, and a leaf $\bar{l}\in l_1(\bar{\Gamma})$ such that $v\in V_{\bar{l}}$, we have $I_{\bar{l}}^x=1$. Thus the execution of $T'$ can be treated as follows: First, pick a decision tree $\Gamma\in \Upsilon_T$ according to $D_T(\Gamma)$, then choose a leaf $l$ labeled $1$ in $\Gamma$ with probability proportional to the fraction of vectors in $V$ that are consistent with the queries to the coordinates in $l$. Finally, query $x$ on the variables in $\var(l)$ and accept if and only if $I_l^x=1$. Therefore, the acceptance probability of $T$ is given by
\[ \Pr[T(x) = 1] = \sum_{\Gamma \in \Upsilon_{T}} D_{T}(\Gamma) \left ( \sum_{l \in l_1(\Gamma)} \frac{| V_l| }{ | V |} \cdot I_l^x \right ). \]

By Claim \ref{equal size 1}, we see that $\rho_x=\Pr[T'(x)=1]$.
\end{proof}


\subsection{Handling real-valued inputs}
In this section, we build a tester for real-valued inputs using a tester for bounded integral inputs. We first show how to handle all integral inputs using a tester for integral inputs from a bounded domain.

\BoundedTargets*
\begin{proof}
If we have a tester $T'(\epsilon,c, s,q)$ for integral inputs, then the same tester can be applied to inputs in ${\cal Z}_d^n$ with the same completeness and soundness parameters and the same query complexity.  %
Given a tester $T(\epsilon, c, s,q)$ for inputs from the domain ${\cal Z}_d^n$, we construct the tester $T'(\epsilon, c, s, q)$ for arbitrary integral inputs as follows: On input $x \in \Z^n$ run 
$T(\epsilon, c, s, q)$ on $w := x \mmod d$, and output the result.

If $x$ is a lattice vector, then from Proposition~\ref{prop:mod-preserves-distance}, we know that $w$ is also a lattice vector, and therefore $T'$ accepts $x$ with probability at least $1-c$. If $x$ is $\epsilon$-far from the lattice, then again from Proposition~\ref{prop:mod-preserves-distance}, we know that $w$ is also $\epsilon$-far from the lattice and $T'$ will accept $x$ with probability at most $s$. We note that the query complexity of $T'$ is identical to that of $T$.
\end{proof}
To address the case of real inputs, we will design a tester for the integer lattice. 

\IntegerLatticeTester*
\begin{proof}
The test queries $O((1/\eps^p)\log(1/s))$ coordinates of the input uniformly at random and accepts iff all the queried coordinates are integral. 

If the input is in the lattice, then all the queried coordinates will be integral, and hence the tester will accept. If the input $w$ is at $\ell_p$ distance at least $\eps \cdot \|1^n\|_p$, then at least $\eps^p n$ coordinates of the input are non-integral. Thus the tester will reject with probability at least $1-s$. 

We note that the tester is a linear test: the test described can be viewed as picking independent uniform random standard basis vectors $e_i\in \Z^n\subseteq L^{\perp}$  (where $e_i$ is the indicator vector of the index $i$), for $i\in [n]$, and  testing if the input $w$ satisfies $\langle w, e_i\rangle\in \Z$.

\end{proof}

\RealTargets*
\begin{proof}
Suppose we have a $\ell_p$-tester $T(\epsilon,c, s,q)$ for integer inputs. We can build a tester $T'$ for real valued inputs as follows:
\begin{enumerate}
\item On input $x\in \R^n$, run the $\ell_p$-tester $\bar{T}(\eps/2,0,s,q_Z)$ for $\Z^n$ from Lemma \ref{lemma:integer-lattice-tester} on input $x$. If the tester $\bar{T}$ rejects, then reject.
\item Else, run $T(\eps/2, c,s,q')$ on $x$ where $q'=q(\eps/2,c,s)$ and reject immediately if any of the coordinates queried are not integers; otherwise output the result of $T$. 
\end{enumerate}
If $x$ is a lattice vector, then the acceptance probability of $T'$ is the same as that of $T$ since the tester used in step 1 is a linear tester. If $d_p(x,L)\ge \epsilon \cdot \|1^n\|_p$, then $d_p(x,\round{x})+d_p(\round{x},L)\ge d_p(x,L)\geq  \epsilon \cdot \|1^n\|_p$ and therefore either $d_p(x,\round{x})$ or  $d_p(\round{x},L)$ is at least $\frac12\eps \cdot \|1^n\|_p$. If $d_p(x,\round{x})\ge \frac{1}{2}\eps \cdot \|1^n\|_p$, then $d_p(x,\Z^n)=d_p(x,\round{x})\ge \frac12\eps \cdot \|1^n\|_p$ and therefore step 1 rejects with probability at least $1-s$. If $d_p(\round{x},L)\ge \frac12\eps \cdot \|1^n\|_p$, then step 2 rejects with probability at least $1-s$. 

The number of queries made by the tester $T'$ is $q(\eps/2,c,s) + O((1/\eps^p)\log{(1/s)})$. We note that since the tester used in step 1 is a non-adaptive linear tester, $T'$ would be a non-adaptive linear tester if $T$ is a non-adaptive linear tester.
\end{proof}


\begin{remark}\label{remark:linearD}
We note that the test described in the proof of Theorem \ref{thm:constrD:test} is not a linear test by definition. We now describe a linear test for the code-formula lattice which is equivalent to the test described in 
Section~\ref{subsec:upper-bound-code-formula}. 

Let $T_p$ denote the tester for $\Z^n$. We assume that each code tester $T_i$ for the code $C_i$ is  linear \cite{BHR05} (i.e $T_i$ queries the input $t_i \in \{0,1\}^n$ at $I_i= \{ i_1, \cdots, i_q \} \subseteq [n]$ coordinates according to some distribution and accepts it if and only if $\inprod{t_i}{v} \equiv 0 \mmod 2$ for every $v \in C^{\bot}_{I_i}$). Consider the following variant of the test, that we call $T_{linear}$, which by definition is a linear test: 
\begin{enumerate}
\item Let each $T_i$ query $I_i \subseteq [n]$ coordinates and let $T_p$ query $I_p$ coordinates. 
\item Let $I = \cup_i I_i \cup I_p$.
\item Accept $t$ if $\inprod{t}{x} \in \Z$ for all $x \in (L^{\bot})_I$ 
\item Reject otherwise. 
\end{enumerate}
Note that the query complexity of $T_{linear}$ is upper bounded by the query complexity of $T$.

If the input is a lattice vector; i.e., $t \in L$, then by definition, the inner product of $t$ with every dual lattice vector would be an integer. Therefore, the test is 1-sided.

We now show that $T_{linear}$ rejects all inputs $t$ which are rejected by $T$ and hence, $T_{linear}$ performs at least as well as $T$. 
If $T$ rejects $t$, then there is some $i \in \{0, 1, \ldots, m-1 \}$ such that $t_i$ is rejected by $T_i$ or $t$ is rejected by $T_p$. We note that each code tester $T_i$ and also $T_p$ are linear. Therefore, if $T_i$ rejects $t_i$, then there are no codewords of $C_i$ which agree with $t_i$ on the coordinates $I_i$ queried by $T_i$. So, for the set $I$ which contains $I_i$, there are no codewords of $C_i$ which agree with $t_i$ on the coordinates in $I$. If $T_p$ rejects $t$, then there is some non-integral coordinate in $I_p$ and hence in $I$.
By definition of the code formula construction, $t$ is a lattice vector if and only if for each $i = 0, \ldots m-1$, $t_i$ is a codeword in $C_i$ and $t \in \Z^n$.
Hence, no lattice vector of $L$  agrees with $t$ on those set of coordinates. Therefore, there exists a dual lattice vector supported on $I$, which does not have an integral inner product with $t$. Therefore, $T_{linear}$ also rejects $t$ (and thus, has at least as good a soundness as the original test $T$).
\end{remark} 

\section{Testing membership of inputs outside the span of the lattice}
\label{sec:inputs-outside-span}

In this section we prove Theorems \ref{thm:lower-bound-for-outside-span}, \ref{thm:upper-bound-for-outside-span}, Corollary~\ref{cor:knapsack-lower-bound-outside-span} and Theorem~\ref{thm:knapsack-tester}. We first recall the definitions.
Let $L $ be a rank $k$ lattice in $\Z^n$. Let $S$ denote the $span(L)$ and
$S^{\bot}$ be the subspace orthogonal to $S$.
Let $U = [u_1, \cdots, u_{n-k}]^T \in \R^{(n-k) \times n}$ be an orthonormal basis for $S^{\perp}$. 
Let $P\subseteq [n]$ be the set of coordinates that support the vectors in $S^{\perp}$ i.e.,
\[ P := \bigcup\limits_{i \in [n-k]} supp(u_i). \]

\thmLowerBoundForOutsideSpan*
\begin{proof}

To show the $\Omega(\lvert P \rvert)$  lower bound, we use Yao's principle: we setup a distribution ${\cal D}$ 
on far inputs such that every deterministic algorithm requires $\Omega(|P|)$ queries to distinguish whether the input is $0\in L$ or is far from $L$.
We define ${\cal D}$  as follows: pick $j$  uniformly at random from $P$, and set $t^{(j)}:= De_j$, where 
$D \geq  \frac{\eps \cdot \|1^n\|_p}{\min \limits_{i, j: u_{i,j}\neq 0} \lvert u_{i,j} \rvert} $. The following claim shows that the distance of each such $t^{(j)}$ from $L$ is at least $\eps \cdot \|1^n\|_p$.

\begin{claim}
$d_p(t^{(j)}, L) \geq \eps \cdot \|1^n\|_p$ for every $j \in P$.
\end{claim}
\begin{proof}
It is sufficient to show that $t^{(j)}$ is far from $S$, since $L\subseteq S$. 
Let $t^{(j)} = t^{(j) \parallel} + t^{(j) \perp}$, 
where $t^{(j) \parallel}$ is the component of $t^{(j)}$ in $S$ and 
$t^{(j)\perp}$ is the component of $t$ in $S^{\perp}$. 
By definition, 
\[t^{(j) \perp} = proj_{S^{\perp}}(t^{(j)}) = \sum\limits_{\ell \in [n-k]} \langle t^{(j)}, u_\ell \rangle u_\ell .\]
Since $U$ is an orthonormal basis of $S^{\perp}$,
\[ \| t^{(j) \perp} \|_p^p = \sum\limits_{\ell \in [n-k]} \lvert \langle t^{(j)}, u_\ell \rangle \rvert^p = \sum\limits_{\ell \in [n-k]} (Du_{\ell, j})^p \geq (\eps \cdot \|1^n\|_p)^p. \]
The last inequality follows from the choice of $D$ and the fact that there exists at least one $u_{\ell,j} \neq 0$ since $j \in P$. 
Therefore, the distance of $t^{(j)}$ from $S$ and hence from $L$, is at least $\eps \cdot \|1^n\|_p$.
\end{proof}

By the choice of the distribution, every deterministic test fails on inputs drawn from ${\cal D}$ with probability $1/|P|$. Thus any randomized test requires $\Omega(|P|)$ queries in order to succeed with constant probability.
\end{proof}

\thmUpperBoundForOutsideSpan*
\begin{proof}
Let $T(\eps, c, s, q)$ be an $\ell_p$-tester for $L$ for inputs in $span(L)$ with query complexity $q=q(\eps)$. We now design a tester $T'(\eps', c', s', q')$ for $L$ for inputs $t = (t_1, t_2, \cdots, t_n) \in \R^n$. 
By making an additional $|P|$ queries, $T'$ can compute the coordinates of the projection of $t$ onto $S$. If $t$ is far from $L$, then either (i) $t$ is far from $S$ or (ii) $t$ is close to $S$ but far from $L$. The coordinates in $P$ would identify if $t$ is far from $S$ and enable rejection. If $t$ is close to $S$ but far from $L$, the tester $T$ would reject the projection of $t$ onto $S$ thus enabling rejection. We now formalize this intuition.

Let $t = (t_1, t_2, \cdots, t_n) \in \R^n$ be the input to the tester $T'$. We compute the projection of $t$ on $span(L)$ by querying all the coordinates in $P$. Let $t^{\perp}$ be the projection of $t$ onto $S^{\perp}$. 
Since $U$ is an orthonormal basis for $S^{\perp}$, we have
\[ t^{\perp} = \sum\limits_{\ell \in [n-k]} \langle t, u_\ell \rangle u_\ell . \]
Each inner product in this expression can be computed using only the coordinates in $P$ and therefore, $t^{\perp}$ can be computed from $t$ by querying just $\lvert P \rvert$ coordinates.  If $\|t^{\perp}\|_p \geq \eps'/2 \cdot \|1^n\|_p$, then $T'$ rejects $t$ immediately.  So we now assume $\|t^{\perp}\|_p < \eps'/2 \cdot \|1^n\|_p$
The projection of $t$ onto $S$ is: 
\[ 
t_{j}^{\parallel} = \left\{
	\begin{array}{ll}
	t_j & \mbox{ if  } j \notin P \\
t_j - t_{j}^{\perp} & \mbox{ if }  j \in P
\end{array}
\right.
\]
Now we run the tester for $T$ on input $t^{\parallel}$ for distance parameter $\eps = \eps'/2$ and accept $t$ if and only if $T$ accepts. 

If $t \in L$, then $t^{\perp} = 0$ and $T$ would accept with probability at least $1-c$. If $d_p(t, L) \geq \eps' \cdot \|1^n\|_p$, then 
\[ d_p(t^{\parallel}, L) \geq \eps' \cdot \|1^n\|_p - d_p(t^{\perp}, S) \geq \eps'/2\cdot \|1^n\|_p = \eps \cdot \|1^n\|_p.  \]

Therefore, $T$ would reject with probability at least $1-s$. 
Finally, note that $q'(\eps') \leq q(\eps'/2)+ |P|$. 
\end{proof}

\subsection{Testing Knapsack Lattices}\label{sec:knapsack-tester}

\CorKnapsackLowerBoundOutsideSpan*
\begin{proof}
We note that $L$ has rank $n-1$ and the vector $(a_1, a_2, \ldots, a_{n-1}, -1)$ generates the subspace orthogonal to $span(L)$, hence the set $P$ of elements in the support of this space has size $|P|=n$, and the lower bound follows from Theorem \ref{thm:lower-bound-for-outside-span}.
\end{proof}

We now prove Theorem  \ref{thm:knapsack-tester}, namely that knapsack lattices can be tested with a constant number of queries if   the inputs come from the span of the lattice. In fact, we will show that testing such lattices simply reduces to testing membership in $\Z^n$.

\ThmKnapsackTester*
\begin{proof}
Let $L=L_{a_1,\ldots,a_{n-1}}$. Let $w\in \text{span}(L)$ denote the input. Any vector $w \in span(L)$ is of the form \[ w = \left(\alpha_1, \cdots, \alpha_{n-1}, \sum_{i=1}^{n-1} a_i \alpha_i \right) \]
for some real values $\alpha_1,\ldots, \alpha_{n-1}$.
Let $w'\in \R^{n-1}$ denote the projection of $w$ on the first $n-1$ coordinates. Let $T_p(\eps',0,s',q')$ denote the $\ell_p$-tester for $\Z^{n-1}$, where $q'= O\left(\left( \frac{1}{\eps'^p}\right) \log\frac{1}{s'}\right)$. 

The tester proceeds as follows: Run the tester $T_p(\eps'=\eps/(M+1)^{1/p},0,s,q=O\left(\left( \frac{M}{\eps^p}\right) \log\frac{1}{s}\right) )$ on input $w'$. Accept if and only if the tester $T_p$ accepts. 

The query complexity of the tester is immediate.  
If $w \in L$, then each coordinate is integral. Therefore the test accepts $w$ with probability 1. We use the following claim to analyze the soundness of the test.

\begin{claim}\label{claim:distance-L-knapsack}
Let $w \in \text{span}(L)$, and $w' = (w_1, \cdots, w_{n-1}) \in \R^{n-1}$ then,
\[ d(w, L)^p \leq (M+1) \cdot d(w', \Z^{n-1})^p \]
\end{claim}
\begin{proof}
Consider the following vector $v \in L$: 
\[ v = (\round{w_1}, \cdots, \round{w_{n-1}}, \sum_{i=1}^{n-1} a_i \round{w_i} ) \]
where $\round{w_i}$ denotes the rounding of $w_i$ the nearest integer. We now upper bound the distance of $w$ from $L$ using this lattice vector $v$. 
\begin{align*}
d(w, L)^p & \leq d(w, v)^p = \| w - v \|_p^p  \\
& = \sum\limits_{i=1}^{n-1} \lvert w_i - \round{w_i} \rvert^p + \lvert w_n - \sum_{i=1}^{n-1} a_i \round{w_i}\rvert^p \\
& = \sum\limits_{i=1}^{n-1} \lvert w_i - \round{w_i} \rvert^p + \lvert \sum_{i=1}^{n-1} a_i w_i - \sum_{i=1}^{n-1} a_i \round{w_i}\rvert^p \\
& \leq \sum\limits_{i=1}^{n-1} \lvert w_i - \round{w_i} \rvert^p +  \sum_{i=1}^{n-1} \lvert a_i (w_i -\round{w_i}) \rvert^p \\
& \leq \sum\limits_{i=1}^{n-1} \lvert w_i - \round{w_i} \rvert^p +  M \sum_{i=1}^{n-1} \lvert w_i -\round{w_i} \rvert^p \\
& = (M+1) \cdot \sum\limits_{i=1}^{n-1} \lvert w_i - \round{w_i} \rvert^p  \\
& = (M+1) \cdot d(w', \Z^{n-1})^p
\end{align*}
\end{proof}

It remains to bound the soundness error probability. If $d(w, L) \geq \eps \|1^n\|_p$, then from Claim~\ref{claim:distance-L-knapsack}, we get that $d(w', \Z^{n-1}) \geq (\eps/(M+1)^{1/p})\|1^n\|_p$. Therefore, the tester $T_{p}$ rejects $w$ with probability at least $1-s$. 
 
\end{proof}

\section{Discussion}
In this paper we defined a notion of local testing for a new family of objects: point lattices. Our results demonstrate connections between lattice testing and the ripe theory of locally testable codes, and bring up numerous avenues for further research (particularly, Questions ~\ref{question-1} and ~\ref{question-2}).

We remark that the notion of being `$\eps$-far' from the lattice may be defined differently than in Definition \ref{defn:local-test}, depending on the application of interest. In particular, 
in applications like IP and cryptography, it is natural to ask for a notion of testing that ensures that scaling the lattice does not change the query complexity. An alternate definition of $\eps$-far based on the \emph{covering radius} of the lattice could be helpful to achieve this property. 
The covering radius of a lattice $L\subseteq \R^n$ (similar to codes) is the largest distance of any vector in $\R^n$ to the lattice. It is trivial to design a tester to verify if a point is in the lattice or at distance more than the covering radius from the lattice (simply accept all inputs). 
In order to have a tester notion where scaling preserves query complexity, we may define a vector as being $\eps$-far from the lattice, if the distance of the vector to every lattice point is at least $\eps$ times the covering radius of the lattice. We note that the covering radius of any {\em integral lattice} is $\Omega(\|1^n \|_p)$. Indeed, the densest possible integral lattice, namely the integer lattice $\Z^n$, has covering radius $(1/2)\|1^n \|_p$, as exhibited by the point $v=(1/2,\ldots, 1/2)\in \R^n$. Thus, by asking the tester to reject points at distance more than $\eps \|1^n \|_p$ in Definition \ref{defn:local-test}, we have settled upon a strong notion of being $\eps$-far from the lattice (i.e., the definition would in particular imply that vectors that are farther than $\eps$ times the covering radius would be rejected by the tester). This definition is essentially equivalent to the current Definition \ref{defn:local-test} if the covering radius of the lattice is $\Theta(n)$. With the modified definition of local testers using covering radius as described above, the equivalent Question 1 is to identify a family of lattices that can be tested using a constant number of queries, achieves constant rate and whose ratio of minimum distance to covering radius is also at least a constant. \\

\noindent{\bf Acknowledgments.}
We thank Chris Peikert for mentioning to us about the potential application to cryptanalysis, and anonymous  reviewers for helpful comments and pointers.

\bibliographystyle{myplain}
\bibliography{PropertyTest-Bib}

\end{document}